\documentclass{amsart}

\usepackage[T1]{fontenc}
\usepackage{lmodern}
\usepackage{amsmath, amsfonts, amssymb}
\usepackage{bookmark}
\usepackage{enumerate}
\usepackage{xifthen}
\usepackage[utf8]{inputenc}
\usepackage[a4paper]{geometry}
\usepackage{amsmath}
\usepackage{amsthm}
\usepackage{mathtools}
\usepackage{booktabs}
\usepackage{setspace}
\newcommand{\mathbold}[1]{{\boldsymbol {#1}}}
\newcommand{\mbf}[1]{{\mathbold {#1 }}}
\newcommand{\LEFTRIGHT}[3]{\left#1 {#3} \right#2 }
\newtheorem{lem}{Lemma}
\newtheorem*{lem*}{Lemma}
\newtheorem{thm}{Theorem}
\newtheorem*{thm*}{Theorem}
\newtheorem{prop}{Proposition}
\newtheorem{cor}{Corollary}
\theoremstyle{remark}

\newtheorem*{rem*}{Remark}
\DeclareMathOperator{\tr}{Tr}

\newcommand{\nor}[1]{\left\| #1 \right\|}
\newcommand{\md}[1]{\vert #1 \vert}

\newcommand{\dd}{\mathrm{d}}

\newcommand{\lio}{\mathrm{o}}
\newcommand{\bigo}{\mathrm{O}}

\newcommand{\scr}{s_{\mathrm{c}}}
\newcommand{\nx}[1][]{\md{\mbf x} \if\relax\detokenize{#1}\relax\else^{#1}\fi}
\newcommand{\ee}[1][]{\mathrm{e}\if\relax\detokenize{#1}\relax\else^{#1}\fi}
\newcommand{\ii}{\mathrm{i}}

\newcommand{\ketf}[3]{\ket{f_{#1 , \mbf{#2} , \mbf{#3}}}}
\newcommand{\braf}[3]{\bra{f_{#1 , \mbf{#2} , \mbf{#3}}}}

\newcommand{\prf}[3]{\pi_{#1, \mbf{#2} , \mbf{#3}}}
\newcommand{\ootp}[1]{\frac{1}{(2\pi)^{#1}}}
\newcommand{\ipu}[3][]{\ifthenelse{\isempty{#1}}{\iint \dd\mbf{#2}\,\dd\mbf{#3}}{\smashoperator{\iint\limits_{#1}\fi }\dd\mbf{#2}\,\dd\mbf{#3}}}
\newcommand{\lkin}{L^{\mathrm{kin}}_d}
\newcommand{\lpot}{L^{\mathrm{pot}}_d}
\newcommand{\lcl}{L^{\mathrm{cl}}_d}
\newcommand{\chrf}{\mathbf 1}

\newcommand{\depsilon}{\epsilon}
\DeclareMathOperator{\hess}{\mathrm{Hess}}
\newcommand{\ball}{\mathbb B}

\newcommand{\scases}[3]{\LEFTRIGHT\lbrace.{\!\!\begin{array}{lcl}
#1 & \text{if} & s<\scr \\
#2 & \text{if} & s = \scr \\
#3 & \text{if} & s>\scr 
\end{array} }}

\newcommand{\scrh}{\frac{\scr}{2}}
\newcommand{\oolt}{\frac{2}{2+d}}
\renewcommand{\log}{\ln}
\renewcommand{\star}{*}

\newcommand{\bra}[1]{\left\langle #1 \right\vert}
\newcommand{\ket}[1]{\left\vert #1 \right\rangle }

\newcommand{\brawv}[2][x]{\bra{\ee^{\ii \mbf{#2} \cdot \mbf{#1}}}}
\newcommand{\ketwv}[2][x]{\ket{\ee^{\ii \mbf{#2} \cdot \mbf{#1}}}}
\newcommand{\osc}[1]{\Omega^{\mathrm{sc}}_{#1 }}
\newcommand{\psc}[1]{\Phi^{\mathrm{sc}}_{#1 }}

\renewcommand{\epsilon}{\varepsilon}

\renewcommand{\mbf}[1]{{\mathbold {#1 }}}

\newcommand{\mbx}[1]{\mbf{#1}}

\makeatletter
\newcommand{\subalign}[1]{
  \vcenter{
    \Let@ \restore@math@cr \default@tag
    \baselineskip\fontdimen10 \scriptfont\tw@
    \advance\baselineskip\fontdimen12 \scriptfont\tw@
    \lineskip\thr@@\fontdimen8 \scriptfont\thr@@
    \lineskiplimit\lineskip
    \ialign{\hfil$\m@th\scriptstyle##$&$\m@th\scriptstyle{}##$\crcr
      #1\crcr
    }
  }
}
\makeatother

\allowdisplaybreaks

\begin{document}

\title[Generalization of Weyl's Law for Relative Traces of Singular Potentials]{A Generalization of Weyl's Asymptotic Formula for the Relative Trace of Singular Potentials}
\author{Jakob Ullmann}
\address{Ludwig Maximilian University of Munich, Department of Mathematics, Theresienstrasse 39, D-80333 Munich, Germany}
\email{jakob.ullmann@gmail.com}

\begin{abstract}
By Weyl's asymptotic formula, for any potential $V$ whose negative part $V_-$ is an $L^{1+d/2}$-function,
\begin{align*}  \tr [-h^2 \Delta + V]_- &= L_d^{\mathrm{cl}} h^{-d} \int \dd\mbf x\,[V]_-^{1+\frac d 2} + \lio (h^{-d})_{h \to 0} ,  \end{align*}
with the semiclassical constant $L^{\mathrm{cl}}_d = 2^{-d} \pi^{-d/2} / \Gamma (2 + \frac d 2)$. In this paper, we show that, even if $[V_1]_-, [V_2]_- \notin L^{1+d/2}$, but the difference $[V_1]_-^{1+d/2}-[V_2]_-^{1+d/2}$ is integrable, then we still have the asymptotic formula 
\[ \tr [-h^2 \Delta + V_1 ]_- - \tr [-h^2 \Delta + V_2 ]_- = L^{\mathrm{cl}}_{d} h^{-d} \int \dd\mbf x\,([V_1]_-^{1+\frac d 2}-[V_2]_-^{1+\frac d 2}) + \lio (h^{-d})_{h\to 0} . \]
 This is a generalization of Weyl's formula in the case that $\tr [-h^2 \Delta + V_1]_-$ and $\tr [-h^2\Delta + V_2]_-$ are seperately not of order $\bigo (h^{-d})$.
\end{abstract}

\maketitle

\setcounter{tocdepth}{2}
\tableofcontents

\section{Introduction}

By the heuristical principle of the semiclassical approximation, which is the core of the Bohr--Sommerfeld quantization rules in `old quantum mechanics,' in the limit $h \downarrow 0$, every phase space cell of unit volume $(2 \pi )^d$ where the classical Hamiltonian function $H^{\mathrm{cl}} (\mbf x, \mbx p ) = \vert\mbf p\vert^2 + V(\mbf x)$, with $\mbf p = h \mbf k$, is negative will hold one negative energy eigenstate of the Schrödinger operator $H = -h^2 \Delta + V(\mbf x)$. Accordingly, the sum of the first moments of the negative eigenvalues $E_i$ would be
\begin{align*}
\sum_i \vert E_i \vert &\cong \smashoperator{\iint\limits_{\mathbb R^d \times \mathbb R^d}} \frac{\dd\mbf x\,\dd\mbf k}{(2 \pi )^{d}}\, [ h^2 \vert\mbf{k}\vert^2 + V(\mbf x) ]_- \\*
&= L^{\mathrm{cl}}_{d} h^{-d}\,\int \dd\mbf x\, [V(\mbf x)]_-^{1 + \frac d 2} ,
\end{align*}
where $L^{\mathrm{cl}}_{d}$ is the \emph{semiclassical constant},
\[ L^{\mathrm{cl}}_{d} = \frac{1}{(4 \pi)^{\frac d 2 } \Gamma (1 + d + \frac d 2 )} . \]
The sum $\sum \vert E_i \vert$ may be written as
\[ \tr [-h^2 \Delta + V ]_- , \]
where the negative part $[-h^2 \Delta + V ]_-$ is understood as the function $- H\,\Theta(- H)$ of the self-adjoint operator $H = -h^2 \Delta + V$, as defined by spectral calculus. In an analogous fashion, by the negative part $[V(\mbf x)]_-$ we always mean a non-negative number.

In fact, this can be proven rigorously as an asymptotic formula, known as the \emph{Weyl asymptotics} (see Weyl \cite{Weyl1911} for the original paper about the number of eigenvalues of the Dirichlet Laplacian, Lieb and Loss \cite{LiebLoss2001} for a coherent state proof of the version for the sum of negative eigenvalues of Schrödinger operators, and Reed and Simon \cite{ReedSimon1978} for a Dirichlet--Neumann bracketing proof of the version for the number of negative eigenvalues of Schrödinger operators).

\begin{thm*}[Weyl asymptotics]
Let $V_+ \in L^1_{\mathrm{loc}} (\mathbb R^d)$ and $V_- \in L^{1 + \frac d 2 }(\mathbb R^d)$, then in the limit $h \downarrow 0$,
\[ \tr [-h^2 \Delta + V]_- = L^{\mathrm{cl}}_{d} h^{-d}\,\int\dd\mbf x\,[V(\mbf x)]_-^{1 + \frac d 2} + \lio (h^{-d}).  \]
\end{thm*}

Closely related to the Weyl asymptotics is the famous \emph{Lieb--Thirring inequality} (see Lieb and Thirring \cite{LiebThirring1975} for the original paper, and Lieb and Seiringer \cite{LiebSeiringer2010} for a review).

\begin{thm*}[Lieb--Thirring inequality]
There exists a universal constant $L_{d} \geq L^{\mathrm{cl}}_{d}$, the \emph{Lieb--Thirring constant}, such that for all $V$ with $V^+ \in L^{1}_{\mathrm{loc}}(\mathbb R^d)$ and $V_- \in L^{1 + \frac d 2 }(\mathbb R^d)$,
\[ \tr [-h^2 \Delta + V]_- \leq L_{d} h^{-d}\,\int\dd\mbf x\,[V(\mbf x)]_-^{1 + \frac d 2} .  \]
\end{thm*}

It is notoriously a major open problem in mathematical physics to prove the Lieb--Thirring inequality for $L_{d} = L^{\mathrm{cl}}_{d}$, which is conjectured to hold true for $d \geq 3$, known as the Lieb--Thirring conjecture.

Both the Weyl asymptotics and the Lieb--Thirring inequalities have generalizations for the sums of lower moments of eigenvalues; in fact, we shall use the Lieb--Thirring inequality for low moments in our proof, see below for its statement.

The semiclassical approximation principle suggests that for two Schrödinger operators $H_1 = -h^2 \Delta + V_1$ and $H_2 = -h^2 \Delta + V_2$, in the limit $h \downarrow 0$,
\begin{equation} \label{eq:conj}\tr [H_1]_- - \tr [H_2]_- = L^{\mathrm{cl}}_{d} h^{-d}\;\int\dd\mbf x\,\big( [V_1(\mbf x)]_-^{1 + \frac d 2 } - [V_2(\mbf x)]_-^{1 + \frac d 2 } \big) + \lio (h^{-d}) . \end{equation}
When $[V_1]_-$ and $[V_2]_-$ are both in $L^{1 + \frac d 2 }$, this follows from the Weyl asymptotics. However, there are important cases in which this condition is not fulfilled, but nevertheless $[V_1(\mbf x)]_-^{1 + \frac d 2 } - [V_2(\mbf x)]_-^{1 + \frac d 2 }$ is integrable. The semiclassical approximation principle suggests that, in those coses, $h^d \tr [-h^2 \Delta  + V_1]_-$ and $h^d \tr [-h^2 \Delta + V_2]_-$ tend to $+ \infty$ in the limit $h \downarrow 0$, but there difference has a finite limit.

This kind of result has an application in mathematical physics in the study of the ground state energy of large Coulomb systems, known as the \emph{Scott correction}. In this context, $\tr [-h^2 \Delta + V_1 ]_-$ is compared with $\tr [-h^2 \Delta + V_2 ]_-$, where $V_1$ is the Thomas--Fermi potential and $V_2 = - \nx[-1] + \mu$, with $\mu > 0$ a constant chemical potential. In the two-dimensional case, the Coulomb potential $- \nx[-1]$ is singular in the sense that it is not in $L^{2}_{\mathrm{loc}}$. See Nam, Portmann and Solovej \cite{Nam2012}, where the authors prove this result in $d = 2$ for potentials with Coulomb-like singularities.

In a similar spirit, Frank, Lewin, Lieb and Seiringer \cite{FrankSeiringer} proved a Lieb--Thirring-like inequality for the comparison of a potential $V-\mu$ with the constant potential $-\mu$ to bound the energy cost to make a hole in the Fermi sea.

In this paper, we give explicit conditions under which we can prove (\ref{eq:conj}) in two and three dimensions. Furthermore, the method generalizes to all $d \geq 2$; only the parameter adjustment is different. Our method generalizes the proof in \cite{Nam2012}.

\bigskip

\noindent\textbf{Acknowledgement.} This work emerged from the author's Master thesis in the graduate program \emph{Theoretical and Mathematical Physics} at LMU and TUM. The author would like to express his gratitude to his supervisor, Prof. Phan Thành Nam, for his patient guidance.

\section{Main Results}

We consider two Schrödinger operators
\begin{align*}
H_1 &= -h^2 \Delta + V_1 \\*
H_2 &= -h^2 \Delta + V_2 
\end{align*}
on the Hilbert space $\mathfrak{H} = L^2(\mathbb R^d )$, in particular we are interested in $d = 3$ and $d = 2$. The Schrödinger operators shall be defined as quadratic forms, where the potentials $V_1, V_2$ are assumed to be $C^1 (\mathbb R^d \setminus \lbrace 0 \rbrace )$-functions. Furthermore, we require them to fulfill the conditions
\begin{align}
\LEFTRIGHT\lbrace. {
  \begin{array}{rcl} 
    \vert V(\mbf x) \vert &\leq& C\text{ for }\nx \geq 1\text{ and }\displaystyle \sup_{\nx \geq L} [V(\mbf x)]_- \rightarrow 0\text{ as }L \to \infty \\ 
    \vert \nabla V (\mbf x) \vert & \leq & \LEFTRIGHT\lbrace. {\begin{array}{ll}
                                   C \nx[-s-1] : & \text{for } \nx \leq 1 \\
                                   C \nx[-S-1] : & \text{for } \nx \geq 1
                                \end{array} }
                                \\
    \vert V_1(\mbf x) - V_2(\mbf x) \vert & \leq & C \nx[-r]\quad \text{for}\; \nx \leq 1 ,
  \end{array}
} \label{eq:conditions}
\end{align}
where an equation involving $V$ means that it shall hold for both $V_1$ and $V_2$. The parameters $s, S, r$ will be specified later.

Note in particular that in the last equation we take the supremum only of the negative parts $\lbrace [V(\mbf x)]_- \rbrace_{\nx \geq L }$, not the function values. The function values $V(\mbf x)$ may, and in many practical applications will, have a positive limit inferior, most notably through the presence of a cutoff (i.e., a positive chemical potential); in these cases, the negative parts $V_-$ have compact support.

We shall assume the parameter $s$ in (\ref{eq:conditions}) to fulfill $1 \leq s < s_{\mathrm{max}} < 2$. The condition $s < 2$ guarantees that $H_1$ and $H_2$ are well-defined as semi-bounded self-adjoint operators with essential spectrum $\sigma_{\mathrm{ess}} (H_1 ), \sigma_{\mathrm{ess}} (H_2) \subset [0, \infty)$. If (\ref{eq:conditions}) is fulfilled with $s < 1$, we can still put $s = 1$ as a non-optimal choice. However, we only optimize our parameters for the case $s \geq 1$, as we are interested in singular potentials, not in obtaining optimal error bounds for regular potentials. 

However, $s < 2$ does not guarantee that $V_- \in L^{1+\frac d 2}(\mathbb R^d)$; that is implied only for
\[ s < \scr := \frac{2 d}{d + 2} , \]
i.e. $\scr = 1$ for $d=2$ and $\scr = \frac 6 5$ for $d = 3$. Our method allows for $s_{\mathrm{max}} > \scr$. Our theorems are thus non-trivial and extend Weyl's asymptotic law.

Our method works in all dimensions $d \geq 2$, but let us restrict the presentation to three and two dimensions for simplicity.

In this paper, we shall prove the following generalization of Weyl's law.

\begin{thm}[Relative Weyl law in three dimensions]
\label{thm:m3}
Let $d = 3$. Assume (\ref{eq:conditions}), $1 \leq s < \frac{62}{45} = 1.377\dots$, $S > \frac 6 5$ and
\begin{equation}  r < \min \LEFTRIGHT\lbrace\rbrace {s , \frac{6 (45 s-62)}{25 (7 s-10)} } . \label{eq:condr3} \end{equation}
Then, there exists an $\eta > 0$ such that, in the limit $h \downarrow 0$,
\[
\tr [- h^2 \Delta + V_1  ]_- - \tr [- h^2 \Delta + V_2  ]_- = \frac{1}{15 \pi^2 h^3} \,\int \dd\mbf x\;\big( [ V_1(\mbf x) ]_-^{\;\frac 5 2} - [ V_2 (\mbf x ) ]_-^{\;\frac 5 2 } \big) +  \bigo (h^{-3 + \eta }).
 \]
\end{thm}

\begin{rem*}
More precisely, we will show that the three-dimensional asymptotic holds for all $\eta < \min \lbrace \eta_{\mathrm{sc}}, \eta_{\mathrm{loc}}, \eta_{\mathrm{cutoff}}  \rbrace$, where
\begin{align*}
 \eta_{\mathrm{sc}} &= \LEFTRIGHT\lbrace. {  \begin{array}{ll}
    \displaystyle \infty : &  \text{if}\;1 \leq s \leq \frac 6 5   \\[.2em]
    \displaystyle \frac{175 r s-250 r-350 s^2+425 s+82}{5 (2-s) (10 s - 5 r + 1)} : & \text{if}\; s > \frac 6 5 .
 \end{array}  }  \\
 \eta_{\mathrm{loc}} &= \LEFTRIGHT\lbrace. {  \begin{array}{ll}
    \displaystyle 2 - \frac{8}{5 (2-r)} : &  \text{if}\;1 \leq s \leq \frac 6 5   \\[.2em]
    \displaystyle \frac{175 r s-250 r-270 s+372}{25 (2 - s) (2 - r)} : & \text{if}\; s > \frac 6 5 .
 \end{array}  }     \\
  \eta_{\mathrm{cutoff}} &= \LEFTRIGHT\lbrace. {  \begin{array}{ll}
    \displaystyle \frac 5 3 \frac{S - \frac 6 5}{S - \frac 2 3}: &  \text{if}\;\frac 6 5 < S \leq \frac {14} 9   \\[.2em]
    \displaystyle \infty : & \text{if}\; S > \frac{14} 9 .
 \end{array}  } 
\end{align*}
This condition implies that $r < \min \lbrace s, \frac 3 2 (2 - s) \rbrace$, which, as we will prove, ensures the well-definedness of the $\dd\mbf x$-integral. 
\end{rem*}

Furthermore, we prove a two-dimensional version.

\begin{thm}[Relative Weyl law in two dimensions]
\label{thm:m2}
Let $d = 2$. Assume (\ref{eq:conditions}), $1 \leq s < \frac{5}{4} = 1.25$, $S > 1$ and
\[  r < \min \LEFTRIGHT\lbrace\rbrace { \frac{5 - 4 s}{4 - 3 s} , \frac{-5s^2 + 8s - 2}{2 - s} } . \]
Then, there exists an $\eta > 0$ such that, in the limit $h \downarrow 0$,
\[
\tr [- h^2 \Delta + V_1  ]_- - \tr [- h^2 \Delta + V_2  ]_- = \frac{1}{8 \pi^2 h^2} \,\int \dd\mbf x\;\big( [ V_1(\mbf x) ]_-^{2} - [ V_2 (\mbf x ) ]_-^{2 } \big) +  \bigo (h^{-2 + \eta }).
\]
\end{thm}

\begin{rem*}
In the two-dimensional case, the asymptotic holds for all $\eta < \min \lbrace \eta_{\mathrm{sc}}, \eta_{\mathrm{loc}}, \eta_{\mathrm{cutoff}}  \rbrace$, where
\begin{align*}
\eta_{\mathrm{sc}} &= -\frac{2 (s - 1)}{2 - s} + \frac{2 (1 - r) (2 s - 1)}{(2 - s)(5 s - 3 r - 1)}  \\*
\eta_{\mathrm{loc}} &= - \frac{2 (s - 1)}{2 - s} + \frac{2 (1 - r) (2 s - 3)}{(2 - s)(2 - r)} \\*
  \eta_{\mathrm{cutoff}} &= \LEFTRIGHT\lbrace. {  \begin{array}{ll}
    \displaystyle \frac 4 3 \frac{S - 1}{S - \frac 2 3}: &  \text{if}\; 1 < S \leq \frac {4} 3   \\
    \displaystyle \infty : & \text{if}\; S > \frac{4} 3 .
 \end{array}  } 
\end{align*}
Again, the condition on $r$ implies that $r < 2-s \leq s$, which ensures the well-definedness of the $\dd\mbf x$-integral.
\end{rem*}

It is worth pointing out that, in fact, we will prove that there exists a constant $K(C, S, s, r, d)$, independent of $V_1$ and $V_2$ themselves, such that, for all $h \in (0, 1]$
\[ \Big\vert  \tr [-h^2 \Delta + V_1 ]_- - \tr [- h^2 \Delta + V_2 ]_- - \lcl h^{-d} \int \dd\mbf x \big( [V_1(\mbf x)]_-^{1+\frac d 2} - [V_2 (\mbf x)]_-^{1 + \frac d 2} \big) \Big\vert \leq K h^{-d + \eta } .
\]
Therefore, the relative Weyl law as we formulated it holds true even if $V_1 , V_2$ have an $h$-dependence, as long as the same constants $s, S, r$ and, most importantly, $C$ can be chosen for all $h$. In particular, this allows for the replacement $V_i \to V_i + \mu (h), \mu (h) \geq 0$.

\section{Preliminaries}
\subsection{Operator-theoretic tools}
We repeat the statement of the Lieb--Thirring inequalities for the reader's convenience. We shall use them only for $\beta > 0$ (i.e. we will not use the CLR bound).

\begin{thm*}[Lieb--Thirring inequalities]
Let $d \geq 1$ and $\beta$ be such that
\[
\beta \in \LEFTRIGHT\lbrace.{ \begin{array}{l l} {}[\frac 1 2 , \infty ) : & \text{for } d = 1 \\ (0, \infty ) : & \text{for } d = 2 \\ {}[0, \infty ) : & \text{for } d \geq 3 \end{array} } .
\] Then, there exists a universal constant $L_{d, \beta} > 0$ such that for all $V_- \in L^{\beta + \frac d 2 } (\mathbb R^d)$,
\[ \tr [-h^2 \Delta + V]_-^\beta \leq L_{d,\beta} h^{-d}\,\int\dd\mbf x\,[V(\mbf x)]_-^{\beta + \frac d 2} .  \]
\end{thm*}

Cf. \cite{LiebSeiringer2010} for a reference.

Recall that a (fermionic) density matrix is a trace-class operator $\gamma$ with $0 \leq \gamma \leq I$, $I$ being the identity operator. By a well-known fact about Hilbert--Schmidt operators on $L^2 (\mathbb R^d)$, such an operator can be written as an integral operator
\[ (\gamma \psi ) (\mbf x) = \int \dd\mbf y\,\gamma(\mbf x, \mbf y)\,\psi (\mbf y) \]
with a kernel $\gamma (\cdot , \cdot ) \in L^2 (\mathbb R^d \times \mathbb R^d)$. Its one-particle density $\rho_\gamma (\mbf x )$ is formally its diagonal part $\gamma (\mbf x , \mbf x)$. Unless $\gamma (\cdot , \cdot )$ is continuous, however, the expression $\gamma (\mbf x, \mbf x)$ is not well-defined, as the diagonal $\lbrace (\mbf x, \mbf x): \mbf x \in \mathbb R^d \rbrace \subset \mathbb R^d \times \mathbb R^d $ is a Lebesgue nullset. Rigorously, spectral theory tells us that there is a representation
\[ \gamma = \sum_{i = 1}^{\infty} \lambda_i \ket{\psi_i}\bra{\psi_i} \]
with $\psi_i \in L^2(\mathbb R^d), \| \psi_i \|_2 = 1$ being its normalized eigenfunctions and $0 \leq \lambda_i \leq 1$ its eigenvalues, and
\[ \tr \gamma = \sum_{i=1}^\infty \lambda_i . \]
Its one-particle density is then defined as the function
\[ \rho_\gamma (\mbf x) = \sum_{i=1}^\infty \lambda_i\,\vert \psi_i (\mbf x) \vert^2 , \]
and it holds rigorously that
\[ \tr \gamma = \int \dd\mbf x\,\rho_\gamma (\mbf x) . \]
The Lieb--Thirring inequalities for $\beta = 1$ are equivalent to the following inequalities, known as the kinetic Lieb--Thirring inequalities.

\begin{thm*}[Kinetic Lieb--Thirring inequalities]
Let $d \geq 2$. There exists a universal constant $K_d$ such that, for any density matrix $\gamma$,
\[ \tr (-\Delta \gamma ) \geq K_d \int \dd\mbf x\,\rho_{\gamma}(\mbf x)^{1+\frac 2 d} . \]
\end{thm*}

Furthermore, we need the IMS (I.M. Sigal, Ismagilov, Morgan, Morgan--Simon, cf.~\cite{IMS}) localization formula,

\begin{thm*}[IMS formula]
Suppose $(\Phi_j )_{1\leq j \leq N} \subset C^2 (\mathbb R^d)$ are functions with
\[ \sum_{j=1}^N \Phi_j(\mbf x)^2 = 1 . \]
Then, in the sense of quadratic forms,
\[ -\Delta = \sum_{j=1}^N \big( \Phi_j (-\Delta ) \Phi_j - \vert \nabla \Phi_j \vert^2 \big) . \]
\end{thm*}

In fact, we shall use a slightly generalized version of it.

\begin{thm*}[Generalized IMS formula]
Suppose $(\Phi_j )_{1\leq j < \infty} \subset C^2 (\mathbb R^d)$ are functions with
\[ \sum_{j=1}^\infty \Phi_j(\mbf x)^2 = 1 , \]
and suppose that for every compact set $K \Subset \mathbb R^d$ only a finite number of $\Phi_j$ are non-vanishing on $K$. Then, in the sense of quadratic forms,
\[ -\Delta = \sum_{j=1}^\infty \Phi_j (-\Delta ) \Phi_j - \sum_{j=1}^\infty \vert \nabla \Phi_j \vert^2 . \]
\end{thm*}

For the proof, it suffices to prove equality on the form core $C_{\mathrm{c}}^\infty (\mathbb R^d)$. But for these functions, the claim follows from the IMS formula by virtue of the local finiteness assumption.

\subsection{Coherent states}
Let $\mbf u, \mbf p \in \mathbb R^d$, $\tau > 0$, and let $g \in C_{\mathrm c}^1 (\mathbb R^d)$ be a non-negative radial function with $\nor{g}_{2} = 1$, and denote
\begin{align*}
g_\tau (\mbf x) &:= \tau^{-\frac  d 2} g(\mbf x / \tau ) \\
f_{\tau, \mbf u, \mbf p} (\mbf x) &:= \ee[-\ii \mbf p \cdot \mbf x ] g_\tau (\mbf x - \mbf u) \\
\pi_{\tau, \mbf u, \mbf p} &:= \ket{f_{\tau, \mbf u, \mbf p}}\bra{f_{\tau, \mbf u, \mbf p}} .
\end{align*}
The overcomplete family $\lbrace \pi_{\tau, \mbf u, \mbf p} \rbrace_{\mbf u, \mbf p}$ is called coherent state family. The scaling factor $\tau$ is a free parameter whose optimal value will have to be determined.

We remind of the following identities about coherent states (see Lieb and Loss \cite{LiebLoss2001}):
\begin{align*}
I &= \frac{1}{(2 \pi)^d} \iint \dd\mbf p\,\dd\mbf u\;\prf \tau u p \\
\tr (- h^2 \Delta \gamma ) &= \frac{1}{(2 \pi)^d} \iint \dd\mbf p\,\dd\mbf u\;h^2\vert\mbf{p}\vert^2\,\tr (\prf \tau u p \gamma ) - h^2 \tau^{-2} \nor{\nabla g}_{2}^2 \tr \gamma \\
\tr [(V \star g_\tau^2 ) \gamma ] &= \ootp d \ipu p u\;V(\mbf u)\,\tr (\prf \tau u p \gamma ) ,
\end{align*}
for $\gamma \in \mathfrak S^1$ a density matrix and $V_- \in L^1 (\mathbb R^d )$. The integral in the first identity is understood in the weak sense.

With the notations
\begin{align*}
L^{\text{pot}}_d &:= \ootp d \omega_d \\
L^{\text{kin}}_d &:= \ootp d \frac{\omega_d}{1+\frac 2 d} \\
L^{\text{cl}}_d &:= L_{\text{pot}} - L_{\text{kin}} = \ootp d \frac{\omega_d}{1+\frac d 2} ,
\end{align*}
$\omega_d = \pi^{d/2} / \Gamma (1 + d/2)$ being the volume of the $d$-dimensional unit ball, and with
\[ \mathfrak M_{\mbf u} := \lbrace \mbf p : h^2 \vert\mbf{p}\vert^2 + V(\mbf u)<0 \rbrace , \]
the following identities, which are useful in the evaluation of integrals arising from the use of coherent states, hold:
\begin{align*}
\ootp d \smashoperator{\int_{\mathfrak M_{\mbf u}}}\dd\mbf p\,h^2 \vert\mbf{p}\vert^2 &= L^{\text{kin}}_d h^{-d} [V(\mbf u)]_-^{1+\frac d 2} \\
\ootp d \smashoperator{\int_{\mathfrak M_{\mbf u}}}\dd\mbf p\,V(\mbf u) &= - L^{\text{pot}}_d h^{-d} [V(\mbf u)]_-^{1+\frac d 2} \\
\ootp d \smashoperator{\int_{\mathfrak M_{\mbf u}}}\dd\mbf p\,[h^2 \vert\mbf{p}\vert^2 + V(\mbf u) ] &= - L^{\text{cl}}_d h^{-d} [V(\mbf u)]_-^{1+\frac d 2} .
\end{align*}

\subsection{Outline of the proof}
We briefly explain the proof strategy. For $\eta < \eta^* := \min \lbrace \eta_{\mathrm{sc}} , \eta_{\mathrm{loc}}, \eta_{\mathrm{cutoff}} \rbrace$ (as defined in the remarks following Theorems~\ref{thm:m3} and \ref{thm:m2}), put
\[ \epsilon \propto \eta^* - \eta > 0, \]
where the constant of proportionality is not of interest here. In the first step, we construct a partition of unity $\lbrace \Phi^{\mathrm{q}} \rbrace \cup \lbrace \psc n \rbrace_{n = -\infty}^N$ for the domains $\Omega^{\mathrm{q}} = \lbrace \nx \lesssim h^{\alpha} \rbrace$ (quantum zone), $\osc n = \lbrace h^{n \epsilon } \lesssim \nx \lesssim h^{(n-1) \epsilon }  \rbrace$ (semiclassical zones). By an application of the (generalized) IMS formula, we obtain
\[  \tr [-h^2 \Delta + V ]_- \sim \tr [\Phi^{\mathrm{q}} (-h^2 \Delta + V ) \Phi^{\mathrm{q}}]_- + \sum_{N \geq n > - \infty } \tr [\psc n (-h^2 \Delta + V ) \psc n]_- .  \]
The error in this approximation, which we shall refer to as the localization error, depends on the choice of the exponent $\alpha$.

In the second step, we compare the quantum terms:
\begin{multline*}  \big\vert \tr [\Phi^{\mathrm{q}} (-h^2 \Delta + V_2 ) \Phi^{\mathrm{q}}]_- - \tr [\Phi^{\mathrm{q}} (-h^2 \Delta + V_1 ) \Phi^{\mathrm{q}}]_- \big\vert \leq \\ \leq
 \| (V_1 - V_2)\,\chrf_{\Omega^{\mathrm{q}}} \|_{1+\frac d 2} \max_{i = 1, 2} \| \rho_{\Phi^{\mathrm{q}} \gamma_i \Phi^{\mathrm{q}}} \|_{1+\frac 2 d} ,  \end{multline*}
where $\gamma_i = \chrf_{(-\infty, 0)}( \Phi^{\mathrm{q}} (-h^2 \Delta + V_i ) \Phi^{\mathrm{q}} ) $. For $s < \scr $, the $L^{1 + \frac 2 d}$-norm of the reduced one-particle density can be bounded with the kinetic Lieb--Thirring inequality. For $s \geq \scr $, we employ a similar technique involving the low-moment Lieb--Thirring inequality: It is well known that Lieb--Thirring inequality for the $\beta$-th moments implies the Lieb--Thirring inequality for the $\beta'$-th moments, for $\beta' > \beta$. The lower-moment Lieb--Thirring inequality will then yield a finite bound, albeit of higher order, even for potentials that are not covered by the higher-moment Lieb--Thirring inequality. Control over the ground state energy is then crucial.

In the third step, we employ a coherent states technique in the semiclassical zones to find
\[  \tr [\psc n (-h^2 \Delta + V ) \psc n ]_- \sim \smashoperator{\iint\limits_{\osc n \times \mathbb R^d}} \frac{\dd\mbf u\,\dd\mbf p}{(2 \pi h)^d}\;\psc n(\mbf u)\,[h^2 \vert\mbf{p}\vert^2 + V(\mbf u) ]_-\,\psc n(\mbf u). \]
The error made in this approximation involves the goodness of the approximation of $V_-$ by $V_- \star g_{\tau_n}^2 $ (for an appropriate convolution kernel $g$), where we put $\tau_n := h^{\beta_n}$.
We introduce a cutoff exponent $\omega$. For the outer semiclassical zones with $n \epsilon < - \omega$, we obtain a more accurate analysis by not paying for the coherent state approximation, but just bounding the $\tr [\psc n (-h^2 \Delta + V ) \psc n ]_-$ by means of the Lieb--Thirring inequality.

Finally, we optimize the parameters $\alpha$ and $\beta_n$ and choose a cutoff exponent $\omega$. We shall choose $\omega$ such that for zones $n \epsilon < -\omega$, bounding the contribution is cheaper than paying for the coherent state analysis. We remark at this point that other treatments of the cutoff are possible than our method, which uses the Lieb--Thirring inequality. For example, one could employ a comparison technique similar to the one we are using in the quantum zone\footnote{This is possible because we chose an infinite partition, where all partition functions have support of finite measure.}.

\section{Proofs}

In the remainder we shall use $C$ for any universal positive constant, i.e. any positive constant that depends on $C, S, s, r, d$. Different lines may correspond to different values of $C$.

At this point, we remind the reader of the definition
\[ \scr := \frac{2 d}{d + 2} . \]

We start with the following observation.

\begin{lem}
The conditions (\ref{eq:conditions}) imply that
\[
 [V_1(\mbf x)]_-, [V_2(\mbf x)]_- \leq \tilde V(\mbf x) := \LEFTRIGHT\lbrace. { \begin{array}{ll} C'\,\nx[- s] & \text{for}\;\nx \leq 1 \\
 C'\,\nx[- S] & \text{for}\;\nx \geq 1 . \end{array} } \]
 \end{lem}
\begin{proof}
In order to see this for $\nx \geq 1$, we integrate the gradient from spatial infinity, i.e. we take some $\lambda > 1$ and write $V(\mbf x)$ as
\begin{align*} V(\mbf x) &= V(\lambda \mbf x) + \int_1^\lambda \dd\lambda'\, \mbf x \cdot \nabla V(\lambda' \mbf x) \\
 &\geq V(\lambda \mbf x) - \nx \int_1^\lambda \dd\lambda'\,\vert \nabla V(\lambda' \mbf x) \vert  \\
 &\geq \liminf_{\lambda \to \infty} V(\lambda \mbf x) - \nx \int_1^\infty \dd\lambda'\, (\nx \lambda' )^{-S - 1} \\
 &\geq - \nx[-S] \int_1^\infty \dd\lambda'\,(\lambda')^{-S - 1} ,
  \end{align*}
where we used the gradient decay condition and
\[ \liminf_{\lambda \to \infty} V(\lambda \mbf x) \geq - \lim_{L\to \infty} \sup \lbrace [V(\mbf x)]_- \rbrace_{\nx \geq L } = 0 . \]
For $\nx > 1$, we write $\mbf x = \nx\,\hat{\mbf x}$ and
\begin{align*} V(\mbf x) &= V(\hat{\mbf x}) - \int_{\nx }^1 \dd\lambda\, \hat{\mbf x} \cdot \nabla V(\lambda \hat{\mbf x}) \\
 &\geq \min_{\vert \mbf x' \vert = 1} V(\mbf x') - \int_{\nx }^1 \dd\lambda\,\vert \nabla V(\lambda \hat{\mbf x }) \vert  \\
 &\geq \text{const} - C \int_{\nx }^1 \dd\lambda\,\lambda^{-s-1} \\
 &\geq - C'\,\nx[-s] . \qedhere
  \end{align*}
\end{proof}

\subsection{A Lieb--Thirring inequality for singular potentials}

We need the following bound on the trace $\tr [-h^2 \Delta - \md{\mbf x}^{-s} + \mu  ]_- $.

\begin{lem}
\label{lem:landaus}
Let $d \geq 2, 0<s<2, \mu>0$. Then, as $h \downarrow 0$,
\[ \tr \big[\!- h^2 \Delta - \nx[-s] + \mu \big]_- = \scases{\bigo (h^{-d})}{\bigo (h^{-d} \md{\log h}  )}{\bigo (h^{-\frac{2s}{2-s}})} . \]
\end{lem}

The essential ingredient is the knowledge of the ground state energy.

\begin{lem}
\label{lem:gse}
For $d \geq 1$ and $0<s<2$, there exists a constant $E>0$ such that
\[ - h^2 \Delta - \nx[-s] \geq - h^{ -\frac {2 s}{2-s}}\,E. \]
\end{lem}

\begin{proof}
Since the potential fulfills $V_- \in L^{\frac d 2} (\mathbb R^d )$, the Schrödinger operator is semi-bounded, i.e.
\[ -h^2 \Delta - \vert \mbf x \vert^{-s} \geq E_h . \]
By scaling $\mbf y := \lambda\,\mbf x, \Delta_{\mbf y} = \lambda^{-2}\Delta_{\mbf x}, \vert \mbf y \vert^{-s} = \lambda^{-s}\vert \mbf x \vert^{-s}$, we obtain
\[ - \Delta_{\mbf y} - \vert \mbf y \vert^{-s} = \lambda^{-s} \left(- \lambda^{s-2}\,\Delta_{\mbf x} - \vert \mbf x \vert^{-s} \right) \]
and thus, by putting $\lambda = h^{2/(2-s)}$,
\begin{align*}
 - h^2 \Delta_{\mbf x} - \vert \mbf x \vert^{-s} &\geq h^{-\frac{2 s}{2-s}}\,E_1 . \qedhere
 \end{align*}
\end{proof}

\begin{lem}
\label{lem:ltsing}
Let $d\geq 2$, $V_+ \in L^{1}_{\mathrm{loc}}$, $V_- \in L^{\epsilon + \frac d 2}$ for some $0<\epsilon<1$, and let $E>0$ such that
\[ -\Delta + V \geq -E . \]
Then, there exist constants $A, B>0$ (depending only on $d$ and $\epsilon$) such that
\[ \tr [-\Delta + V]_- \leq A \smashoperator{\int\limits_{V(\mbf x) \geq - \frac E 2}} \dd \mbf x\,[V(\mbf x)]_-^{1 + \frac d 2} + B\,E^{1-\epsilon} \smashoperator{\int\limits_{V(\mbf x)< - \frac E 2}} \dd \mbf x\,[V(\mbf x)]_-^{\epsilon + \frac d 2} .  \]
\end{lem}
\begin{proof}
We put $\epsilon := 2 \beta$.

By the Lieb--Thirring inequality,
\begin{align*}
\tr [-\Delta + V]_- &= \int_0^E \dd\lambda\,N_{< - \frac \lambda 2} (-\Delta + V + \tfrac \lambda 2 ) \\
&\leq 2^\beta \int_0^E \frac{\dd \lambda}{\lambda^\beta} \tr [ -\Delta + V + \tfrac \lambda 2 ]_-^\beta \\
&\leq 2^\beta L_{d,\beta} \int_0^E \frac{\dd \lambda}{\lambda^\beta} \smashoperator{\int\limits_{V(\mbf x) \leq - \frac \lambda 2}} \dd\,\mbf x\,\md{[V(\mbf x)]_- - \tfrac \lambda 2}^{\beta + \frac d 2} \\
&= 2^\beta L_{d,\beta} \Big(     \smashoperator{\int\limits_{V(\mbf x) \geq - \frac E 2}} \dd \mbf x \smashoperator{\int\limits_0^{2 \md{V(\mbf x)}}} \frac{\dd \lambda}{\lambda^\beta } \md{[V(\mbf x)]_- - \tfrac \lambda 2 }^{\beta + \frac d 2 }\; + \\
&\hphantom{2 L_{d,\beta} \Big(}
        +\;       \smashoperator{\int\limits_{V(\mbf x)<- \frac E 2}} \dd \mbf x \int_0^{E} \frac{\dd \lambda}{\lambda^\beta } \md{[V(\mbf x)]_- - \tfrac \lambda 2 }^{\beta + \frac d 2 }  \Big) \\
&= 2 L_{d,\beta} \Big( \smashoperator{\int\limits_{V(\mbf x) \geq - \frac E 2}} \dd \mbf x\,[V(\mbf x)]_-^{1+\frac d 2}\;\; \int_0^{1} \frac{\dd \nu}{\nu^\beta } (1 - \nu )^{\beta + \frac d 2 } 
        \;+ \\
        &\hphantom{2 L_{d,\beta} \Big(} +\;      \smashoperator{\int\limits_{V(\mbf x)<- \frac E 2}} \dd \mbf x\,[V(\mbf x)]_-^{1+\frac d 2}\;\;\smashoperator{\int_0^{E/2\md{V}}} \frac{\dd \nu}{\nu^\beta } (1-\nu )^{\beta + \frac d 2 }  \Big) ,
\end{align*}
where we substituted $\lambda =: 2 \md{V(\mbf x)} \nu$ in the last step. We now treat the first (the classical) and the second (the quantum) $\dd\mbf x$-integral separately.

In the classical integral, the $\dd\nu$-integral has $\mbf x$-independent boundaries. The $\dd\nu$-integral is a finite constant; it equals $\mathrm{B} (1 - \beta , 1 + \beta + \tfrac d 2)$, where $\mathrm B (x, y) = \Gamma(x) \Gamma(y) / \Gamma (x+y)$ is the Euler Beta function. Therefore, the classical part may be expressed as
\[
2\mathrm{B} (1 - \beta , 1 + \beta + \tfrac d 2) L_{d,\beta} \,\smashoperator{\int\limits_{V(\mbf x) \geq - \frac E 2}} \dd \mbf x\,[V(\mbf x)]_-^{1 + \frac d 2} .
\]

In order to bound the quantum integral, we use that, for $0<\nu<E / (2 \md{V(\mbf x)} )$,
\[ \frac{1}{\nu^\beta} \leq \left( \frac{E}{2 \md{V(\mbf x)}}\right)^{1-2\beta} \frac{1}{\nu^{1 - \beta}} . \]
After untertaking this modification, we let the $\dd \nu$-integral run to $1$ and thereby obtain the upper bound
\[
2^{2\beta} \mathrm{B} (\beta , 1 + \beta + \tfrac d 2) L_{d,\beta}\,E^{1-2\beta } \,\smashoperator{\int\limits_{V(\mbf x) \geq - \frac E 2}} \dd \mbf x\,[V(\mbf x)]_-^{2 \beta + \frac d 2} . \qedhere
\]
\end{proof}

We now apply this bound to the Schrödinger operator ($\mu$ being a positive constant)
\[ H = - h^2 \Delta + (-\nx[-s] + \mu ) . \]
\begin{proof}[Proof of Lemma~\ref{lem:landaus}]
In order for the quantum integral to be well-defined, it is necessary that we put
\[ \epsilon<d \tfrac{2-s}{2s}. \]
Our analysis in Lemma~\ref{lem:gse} provides us with the ground state energy
\[ - \big( C\,h^{- \frac{2 s}{2-s}} - \mu \big) . \]
Putting
\[ h^2 E = 2 \big( C\,h^{- \frac{2 s}{2-s}} - \mu \big) , \]
we have, for $\mbf x$ in the domain of the quantum integral,
\[ -\nx[-s] + \mu <-\tfrac 12 h^2 E , \]
i.e.
\[ \md{\mbf x}<C h^{\frac{2}{2-s}} . \]

The classical integral (multiplied with $h^2$) is
\begin{align*}
h^2 \smashoperator{\int\limits_{\substack{\md{\mbf x} \geq C h^{\frac{2}{2-s}} \\ \nx \leq \mu^{-\frac 1 s} }}} \dd \mbf x\,h^{-(2+d)}\md{-\nx[-s] + \mu }^{1 + \frac d 2} &\leq h^2 \smashoperator{\int\limits_{\substack{\md{\mbf x} \geq C h^{\frac{2}{2-s}} \\ \nx \leq \mu^{-\frac 1 s} }}} \dd \mbf x\,h^{-(2+d)}\,\nx[-s \left(1 + \frac d 2 \right)]  \\
&= \scases{\bigo (h^{-d})}{\bigo (h^{-d} \md{\log h})}{\bigo (h^{-\frac{2s}{2-s}})} .
\end{align*}
On the other hand, the quantum integral is
\begin{align*}
h^2\,E^{1-\epsilon} \smashoperator{\int\limits_{\md{\mbf x} \leq C h^{\frac{2}{2-s}}}} \dd \mbf x\,h^{-(d + 2 \epsilon )}\md{-\nx[-s] + \mu }^{\epsilon + \frac d 2} &= \bigo ( h^{2 -\frac{4}{2-s} - d 
+ \epsilon ( \frac{4}{2-s} - 2 ) + d - \frac{2 s}{2 - s} \epsilon } ) \\[-2em]
&= \bigo (h^{- \frac{2 s}{2 - s}}) . \qedhere
\end{align*}
\end{proof}

\subsection{$L^p$ bounds for the one-body densities}

\begin{lem}
\label{lem:traceest}
Let $\Omega \subset \mathbb R^d$ have finite Lebesgue measure, $0 \leq \Phi (\mbf x) \leq 1$ and $V(\mbf x)$ be real-valued functions supported on $\Omega$, $\gamma \in \mathfrak L(L^2(\mathbb R^d ))$ with\footnote{Note that $\gamma$ is not required to be trace class, i.e. a density operator.} $0 \leq \gamma \leq I$, and assume
\[ \tr [(-h^2 \Delta + V) \Phi \gamma \Phi ] \leq 0 . \]
\begin{enumerate}[(i)]
\item  Let $V \in L^{1+\frac d 2}(\Omega )$. Then, $\Phi \gamma \Phi$ is trace class and there exists a constant $C = C(d)$ (i.e. independent of $h$) such that 
\begin{align*} \|\rho_{\Phi \gamma \Phi}\|_{{1+\frac 2 d}} &\leq C h^{-d} \nor{V_- }_{{{1+\frac d 2}}}^{\frac d 2}  \\
 \tr (\Phi \gamma \Phi ) &\leq C h^{-d} \nor{V_- }_{{1+\frac d 2}}^{\frac d 2} \md{\Omega}^{\oolt} . \end{align*}
In particular, if $V_- \in L^\infty (\Omega )$,
\begin{align*}
 \| \rho_{\Phi \gamma \Phi}\|_{{1+\frac 2 d}} &\leq C h^{-d} \nor{V_- }_{{\infty}}^{\frac d 2}\,\md{\Omega}^{\frac{d}{2+d}}  \\
 \tr (\Phi \gamma \Phi ) &\leq C h^{-d} \nor{V_- }_{\infty}^{\frac d 2}\,\md{\Omega} .
 \end{align*}
\item  Let
\[ [ V(\mbf x) ]_- \leq C_0 \nx[-s]   \] for some $0<s<2$ and $C_0>0$. Then, $\Phi \gamma \Phi$ is trace class and there exists a constant $C = C(C_0, s, d)$ such that for all $h < 1$
\begin{align*} \|\rho_{\Phi \gamma \Phi }\|_{{1+\frac 2 d}} &\leq	\scases{ C h^{-d} ( \md{\Omega} + 1)^{\scrh}   }{ C h^{-d} ( \md{\Omega} + \vert \log h \vert )^{\scrh}   }{ C h^{-d} \big(\md{\Omega} + h^{-\frac{4 (s-\scr)}{(2-s)(2-\scr )}} \big)^{\scrh}   } .  
\end{align*}
\end{enumerate}
\end{lem}

\begin{proof}
\begin{enumerate}[(i)]
\item For $a>0$,
\[  0 \geq \tr \left[ ( -\tfrac 12 h^2 \Delta + a ) \Phi \gamma \Phi  \right] + \tr \left[ ( -\tfrac 12 h^2 \Delta - (V_- + a)\,\chrf_{\Omega} ) \Phi \gamma \Phi  \right] . \]
Putting $a = 1$ and using $-\Delta \geq 0$, we infer that $\tr (\Phi \gamma \Phi ) < \infty$, i.e. $\Phi \gamma \Phi$ is trace class, because the first term on the right hand side is non-negative and the second is bounded below by the Lieb--Thirring inequality.

Moreover, putting $a = 0$,
\begin{align*}
\tfrac 12 h^2 \tr [-\Delta\,(\Phi\gamma\Phi ) ] &\leq - \tr \left[ ( - \tfrac 12 h^2 \Delta - V_- ) \Phi \gamma \Phi \right] \\
&\leq C h^{-d}\,\int_\Omega \dd\mbf x\,[V(\mbf x)]_-^{1+\frac d 2} .
\end{align*}
The bound on $\| \rho_{\Phi \gamma \Phi} \|_{1+2/d}$ follows by the kinetic Lieb--Thirring inequality, which is applicable because $\Phi \gamma \Phi$ is trace class. The bound on the trace $\tr (\Phi \gamma \Phi )$ then follows by Hölder's inequality.

\item The argument is overall similar to case (i). However, in the second step the Lieb--Thirring inequality does not lead to a finite bound, because $V_-$ is not in $L^{1+\frac d 2}$. Instead, we use the bound from Lemma~\ref{lem:landaus}:
\begin{align*}
\tfrac 12 h^2 \tr [-\Delta\,(\Phi\gamma\Phi) ] &\leq - \tr \left[ ( - \tfrac 12 h^2 \Delta - V_- ) \Phi \gamma \Phi \right] \\
&\leq - \tr \left[ ( - \tfrac 14 h^2 \Delta - C_0 \md{\mbf x}^{-s} + 1 ) \Phi\gamma \Phi \right] - \\
&\hphantom{\leq\;} - \tr \left[ ( - \tfrac 14 h^2 \Delta - 1 ) \Phi \gamma \Phi \right] \\
&\leq \scases{C h^{-d} (1+ \md{\Omega})}{C h^{-d} (\md{\log h}+ \md{\Omega})}{C h^{-\frac{2s}{2-s}} + h^{-d} \md{\Omega}} ,
\end{align*}
where in the last step we used Lemma~\ref{lem:landaus} on the first and the Lieb--Thirring inequality on the second term (multiplying a characteristic function $\chrf_\Omega$ to the potential). \qedhere \end{enumerate}
\end{proof}

\subsection{Localization}
We assume a positive constant $\alpha > 0$ and a (small) positive constant $\depsilon > 0$ to be fixed, and $\depsilon$ be such that $\alpha \in \depsilon \mathbb N$. We define
\begin{align*}
N &:= \alpha / \depsilon \\
\theta_n &:= n \depsilon \qquad (n \in \mathbb Z, N \geq n > - \infty ) .
\end{align*} 

Let $\varphi \in C_{\mathrm{c}}^2 (\mathbb R^d )$ be a radial function $0 \leq \varphi \leq 1$ with $\varphi (\mbf x) = 1$ for $\md{\mbf x} \leq 1$ and $\varphi (\mbf x) = 0$ for $\md{\mbf x} \geq 2$.

We define the following localization functions:
\begin{align*}
\Phi^{\mathrm{q}} (\mbf x) &:= \varphi ( h^{-\alpha} \mbf x) \\
\Phi^{\mathrm{sc}}_{n} (\mbf x) &:= \varphi (h^{-\theta_{n-1}} \mbf x ) \prod_{N \geq j \geq n} \sqrt{1 - \varphi^2 ( h^{- \theta_j} \mbf x)} \quad \text{for}\;N \geq n>-\infty \\
&= \varphi (h^{-\theta_{n-1}} \mbf x ) \sqrt{1 - \varphi^2 (h^{-\theta_n} \mbf x)} \quad \text{(for $h$ small enough)} 
\end{align*}
We denote  $\Omega^{\mathrm q} := \operatorname{supp} \Phi^{\mathrm{q}}$, and so forth. 

These functions fulfill the conditions of the IMS localization formula, i.e.
\[ (\Phi^{\mathrm{q}}  )^2 + \sum_{N \geq n>-\infty} (\Phi^{\mathrm{sc}}_{n} )^2 = 1 . \]

To streamline our notation, we use the shortcut $\sum_n \dots \Phi_n \dots$ to denote summation over all localization functions (quantum and semiclassical).

Furthermore,
\begin{align*}
\Phi^{\mathrm q} (\mbf x) &= \LEFTRIGHT\lbrace. { \begin{array}{ll} 0 & \text{for}\;\nx \geq 2 h^\alpha \\
          1 & \text{for}\;\nx \leq h^\alpha \end{array} } \\
\Phi^{\mathrm{sc}}_{n} (\mbf x) &= \LEFTRIGHT\lbrace. { \begin{array}{ll} 0 & \text{for}\;\nx \leq h^{\theta_n}\;\text{or}\;\nx \geq 2 h^{\theta_{n-1}} \\ 1 & \text{for}\; 2 h^{\theta_n} \leq \nx \leq h^{\theta_{n-1}} \end{array} }  \\
\end{align*}
Finally, we observe that
\begin{align*}
\md{\nabla \Phi^{\mathrm{q}} (\mbf x) }^2 &\leq \nor{\nabla \varphi}_\infty^2 h^{- 2 \alpha}\,\chrf_{\nx \leq 2 h^\alpha} \\
\md{\nabla \Phi^{\mathrm{sc}}_{n} (\mbf x) }^2 &\leq 2 \nor{\nabla \varphi}_{\infty}^2 h^{- 2 \theta_n}\,\chrf_{h^{\theta_n} \leq \nx \leq 2 h^{\theta_{n-1}}} , 
\end{align*}
and hence
\[ \sum_j \md{\nabla \Phi_j (\mbf x)}^2 \leq C h^{- 2 \theta_n} \quad \text{for}\;\mbf x \in \Omega^{\mathrm{sc}}_{n} . \]
Similarly, for all multiindices $\mbf \beta \in \mathbb N_0^d$,
\[
\md{\partial^{\mbf \beta } \Phi^{\mathrm{sc}}_{n} (\mbf x) } \leq 2^{\vert \mbf \beta \vert}h^{- \vert \mbf \beta \vert\,\theta_n}\, \Big(\!\max_{\vert \mbf \beta' \vert \leq \vert \mbf \beta \vert}\| \partial^{\mbf \beta'} \varphi \|_{\infty}\!\Big) \,\chrf_{h^{\theta_n} \leq \nx \leq 2 h^{\theta_{n-1}}} .
\]

\begin{lem}[Localization error]
\label{lem:localization}
Let $V(\mbf x)$ be like specified above. Then there is a constant $A>0$ (depending only on $d$) such that, for $\depsilon$ sufficiently small,
\begin{align*}
\tr [-h^2 \Delta + V ]_- &= \tr [\Phi^{\mathrm{q}} (-h^2 \Delta + V ) \Phi^{\mathrm{q}}]_- + \\*
&\qquad + \sum_{N \geq n>-\infty} \tr [\Phi^{\mathrm{sc}}_{n} (-h^2 \Delta + V ) \Phi^{\mathrm{sc}}_{n}]_-\;+ \\*
&\qquad + \scases{ 
      \bigo (h^{-d + 2(1-\alpha) + \scr \alpha }) }
   {  \bigo \big(h^{-d + 2(1-\alpha) + \scr \alpha - A \depsilon } \big)  }
   {  \bigo \big(h^{-d + 2 (1-\alpha) + \scr \alpha - \frac{4 (s- \scr )}{(2-s)(2-\scr )} } \big)  } .
\end{align*}
\end{lem}

\begin{proof}
We start with the upper bound. Choosing
\[ \gamma := \sum_{j} \Phi_j \chrf_{(-\infty, 0)}[\Phi_j (-h^2 \Delta + V ) \Phi_j ] \Phi_j \leq \sum_{j} \Phi_j^2 = 1 , \]
we infer
\begin{align*}
-\tr [-h^2 \Delta + V ]_- &\leq \tr [(-h^2 \Delta + V ) \gamma ] = -\sum_{j} \tr [\Phi_j (-h^2 \Delta + V ) \Phi_j ]_- .
\end{align*}
This proves the upper bound.

For the lower bound, we find with the IMS formula and the min-max principle that
\begin{align*}
-\tr [-h^2 \Delta + V ]_- &\geq -\tr \Big[ \sum_{j} \Phi_j (-h^2 \Delta + V - \sum_k \md{\nabla \Phi_k}^2 ) \Phi_j \Big]_-  \\
&\geq -\sum_{j} \tr [ \Phi_j (-h^2 \Delta + V - \sum_{k} \md{\nabla \Phi_k}^2 ) \Phi_j ]_- \\
&\geq -\tr [ \Phi^{\mathrm q} (-h^2 \Delta + V - C h^{2 (1 - \alpha )} ) \Phi^{\mathrm{q}} ]_- - \\
&\qquad -  \sum_{N \geq n>-\infty} \tr [ \Phi^{\mathrm{sc}}_{n} (-h^2 \Delta + V - C h^{2(1 - \theta_n )} ) \Phi^{\mathrm{sc}}_{n} ]_- .
\end{align*}

We put
\begin{align*} \gamma^{\mathrm{q}} &:= \chrf_{(-\infty, 0 )} [ \Phi^{\mathrm{q}}  (-h^2 \Delta + V - C h^{2(1-\alpha)}) \Phi^{\mathrm{q}} ]    \\
\gamma^{\mathrm{sc}}_n &:= \chrf_{(-\infty, 0 )} [ \psc n  (-h^2 \Delta + V - C h^{2(1-\theta_n )}) \psc n ] .
 \end{align*}
For the quantum zone,
\begin{multline*}
-\tr [\Phi^{\mathrm{q}} (-h^2 \Delta + V - C h^{2 (1 - \alpha)} ) \Phi^{\mathrm{q}}]_- = -\tr [\Phi^{\mathrm{q}} (-h^2 \Delta + V - C h^{2 (1 - \alpha)} ) \Phi^{\mathrm{q}} \gamma^{\mathrm q}] \\
\geq -\tr [\Phi^{\mathrm{q}} (-h^2 \Delta + V ) \Phi^{\mathrm{q}} \gamma^{\mathrm q}] - C h^{2 (1 - \alpha )} \nor{\rho_{\Phi^{\mathrm{q}} \gamma^{\mathrm{q}} \Phi^{\mathrm{q} } } }_{1 + \frac 2 d} \md{\Omega^{\mathrm{q}}}^{\frac{2}{2+d}} .
\end{multline*}
Similarly, for the semiclassical zones,
\begin{multline*}
-\tr [\psc n (-h^2 \Delta + V - C h^{2 (1 - \theta_n )} ) \psc n]_- = -\tr [\psc n (-h^2 \Delta + V  - C h^{2 (1 - \theta_n )} ) \psc n \gamma^{\mathrm{sc}}_n ] \\
\geq -\tr [\psc n (-h^2 \Delta + V ) \psc n \gamma^{\mathrm{sc}}_n ] - C h^{2 (1 - \theta_n )} \tr (\psc n \gamma^{\mathrm{sc}}_n \psc n ) .
\end{multline*}
In the quantum zone $\Omega^{\mathrm{q}}$ we obtain
\begin{align*}
h^{2(1 - \alpha)} \|\rho_{\Phi^{\mathrm{q}} \gamma^{\mathrm{q}} \Phi^{\mathrm{q} } } \|_{1 + \frac 2 d} \md{\Omega^{\mathrm{q}}}^{\frac{2}{2+d}} &= \scases{ 
      \bigo (h^{-d + 2(1-\alpha) + \scr \alpha }) }
   {  \bigo \big(h^{-d + 2(1-\alpha) + \scr \alpha } \md{\log h}^{\frac \scr 2}\big)  }
   {  \bigo \big(h^{-d + 2 (1-\alpha) + \scr \alpha - \frac{4 (s- \scr )}{(2-s)(2-\scr )} } \big)  }
\end{align*}
Note that the bound on $\|\rho_{\Phi^{\mathrm{q}} \gamma^{\mathrm{q}} \Phi^{\mathrm{q} } } \|_{1 + \frac 2 d}$ involves
\[ h^{-d} \big\| [ V - C h^{2 (1-\alpha )} ]_-\,\chrf_{\Omega^{\mathrm q}} \big\|_{1+\frac d 2}^{\frac d 2} .\]
By elementary analysis, this term is of order
\[ \bigo \big( h^{-d} \nor{V_-\,\chrf_{\Omega^{\mathrm{q}}}}_{1+\frac d 2}^{\frac d 2} + h^{-d} \| h^{2 (1 - \alpha )}\,\chrf_{\Omega^{\mathrm{q}}} \|_{1+\frac d 2}^{\frac d 2} \big) = \bigo \big( h^{-d} \| \tilde V_-\,\chrf_{\Omega^{\mathrm{q}}} \|_{1+\frac d 2}^{\frac d 2} \big) . \]
The error terms on the right hand side above were obtained by evaluation of this term.

The same applies to the analysis of the semiclassical zones: The contribution from $h^{2 (1 - \theta_n )}$ is of lower order, because, for all $n$,
\[ h^{2 (1 - \theta_{n} )}<\min_{\Omega^{\mathrm{sc}}_{n}} \vert \tilde V \vert \propto h^{-S \theta_{n-1} } , \]
provided $h$ and $\depsilon$ are sufficiently small, since we only consider $S \leq \scr<2$.

For the inner (i.e. $n\geq 1$) semiclassical zones $\Omega^{\mathrm{sc}}_n $ we obtain, with a constant $C > 0$ independent of the zone index $n$,
\begin{multline*}
h^{2(1 - \theta_n )}\;\tr (\psc n \gamma^{\mathrm{sc}}_n \psc n ) \leq C\,h^{-d + 2 (1 - \theta_n )}\;\big( \max_{\Omega^{\mathrm{sc}}_{n}} \vert \tilde V \vert \big)^{\frac d 2}\;\vert \osc n \vert \leq  \\
\leq C h^{-d + 2 (1 - \theta_n ) - \frac d 2 s \theta_n + \theta_{n-1} d }  \leq C  h^{-d + 2 + \theta_n (-2 + \frac d 2 (2 - s)) - A \depsilon } ,
\end{multline*}
where $A > 0$ is a constant. The same applies for the outer semiclassical zones with $S$ instead of $s$.

The condition $S, s \geq 2 (1 - \frac 2 d )$ ensures that the coefficient $- 2 + \frac d 2 (2 - s)$ (or $- 2 + \frac d 2 (2 - S)$, respectively) is non-positive, i.e. from all semiclassical zones, the innermost zone, verging on the quantum zone, produces the most critical localization error.

The total semiclassical localization error is finite because the error terms form a geometric series,
\[ \sum_{-n_0 \geq n > - \infty} C  h^{-d + 2 + \theta_n (-2 + \frac d 2 (2 - S)) - A \depsilon }  = \bigo ( h^{-d + 2 + \theta_{-n_0} (-2 + \frac d 2 (2 - S)) - A \depsilon } ) . \]

In conclusion, for $S , s \geq 2 (1 - \frac 2 d)$ there is a constant $A>0$ such that the total localization error is\footnote{For $s \neq \scr$ this is the error in the quantum zone, for $s = \scr$ it is the error in the innermost semiclassical zone.}
\[ \scases{ 
      \bigo (h^{-d + 2(1-\alpha) + \scr \alpha }) }
   {  \bigo (h^{-d + 2(1-\alpha) + \scr \alpha - A \depsilon } )  }
   {  \bigo \big( h^{-d + 2 (1-\alpha) + \scr \alpha - \frac{4 (s- \scr )}{(2-s)(2-\scr )} } \big)  } . \qedhere \]
\end{proof}

\subsection{Comparison in the quantum zone}

\begin{prop}
\label{prop:diffint}
Let $0 \leq \beta \leq 1$ and denote
\[ W_\beta (\mbf x) := [V_1(\mbf x)]_-^{\beta + \frac d 2} - [V_2 (\mbf x)]_-^{\beta + \frac d 2} . \]
There exists a constant $C_\beta > 0$ such that, for $\nx \leq 1$,
\[  \vert W_\beta (\mbf x) \vert \leq C_\beta\,\nx[s(\frac d 2 + \beta - 1 ) - r] .  \] 
In particular, $W_\beta$ is integrable on $\lbrace \nx \leq 1 \rbrace$.
\end{prop}

\begin{proof}
We may assume without restriction that $0 \geq V_2 (\mbf x) \geq V_1(\mbf x)$. Then, by the  non-negativity and monotonicity of $\nx[\frac d 2 + \beta - 1]$,
\begin{align*}
\vert W_\beta (\mbf x) \vert &= C {\int_{[V_2(\mbf x)]_-}^{[V_1 (\mbf x)]_-}}\dd v\,v^{\frac d 2 + \beta - 1} \\
&\leq C {\int_0^{C \nx[-r]}} \dd v\,\big( [V_2(\mbf x)]_- + v \big)^{\frac d 2 + \beta - 1} \\
&\leq C {\int_0^{C \nx[-r]}} \dd v\,\big( C \nx[-s] + v \big)^{\frac d 2 + \beta - 1} \\
&\leq C \nx[s (\frac d 2 + \beta - 1) - r] . \qedhere
\end{align*}
\end{proof}

The following estimate follows by integrating the above estimate on $W_1$ over $\Omega^{\mathrm q}$.

\begin{cor}[Integral error]
\label{cor:quantum-int}
The error from the evaluation of the integral is
\[
h^{-d}\,\int \dd\mbf x\; \big( [ V_1(\mbf x) ]_-^{1 + \frac d 2} - [ V_2 (\mbf x ) ]_-^{1 + \frac d 2 } \big)\,(\Phi^{\mathrm{q}})^2 (\mbf x) = \bigo (h^{-d + \alpha (\frac d 2 (2 - s) - r)}) .
\]
\end{cor}

The following error term we shall refer to as the \emph{quantum error}.
\begin{lem}[Quantum error]
\label{lem:compone}
In the quantum zone,
\begin{multline*} \tr [\Phi^{\mathrm{q}} (-h^2 \Delta  + V_2  ) \Phi^{\mathrm{q}}]_- - \tr [\Phi^{\mathrm{q}} (-h^2 \Delta  + V_1 ) \Phi^{\mathrm{q}}]_- \\
 = \scases{\bigo ( h^{-d + \alpha (\scr - r )})  }{  \bigo ( h^{-d + \alpha (\scr - r )} \md{\log h}^{\scrh})  }{ \bigo ( h^{-d + \alpha (\scr - r ) - \frac{2\scr (s-\scr )}{(2-s)(2-\scr )}})  } . \end{multline*}
\end{lem}

\begin{proof}
Since $V_1$ and $V_2$ are interchangeable, it is enough to prove the upper bound. Let
\[ \gamma := \chrf_{(-\infty, 0)} [\Phi^{\mathrm{q}} (-h^2 \Delta + V_1  ) \Phi^{\mathrm{q}} ] . \]
Then,
\begin{align*}
 &-\tr [\Phi^{\mathrm{q}} (-h^2 \Delta + V_1  ) \Phi^{\mathrm{q}}]_- = \tr [\Phi^{\mathrm{q}} (-h^2 \Delta  + V_1 ) \Phi^{\mathrm{q}} \gamma ] \\
 &\qquad \geq \tr [\Phi^{\mathrm{q}} (-h^2 \Delta  + V_2  ) \Phi^{\mathrm{q}} \gamma ] - C \tr (\Phi^{\mathrm{q}} \nx[-r] \Phi^{\mathrm{q}} \gamma ) \\
 &\qquad \geq -\tr [\Phi^{\mathrm{q}} (-h^2 \Delta  + V_2 ) \Phi^{\mathrm{q}} ]_- - C \tr (\Phi^{\mathrm{q}} \nx[-r] \Phi^{\mathrm{q}} \gamma )  ,
\end{align*}
and
\begin{align*}
\tr (\Phi^{\mathrm{q}} \nx[-a] \Phi^{\mathrm{q}} \gamma) &= \int_{\Omega^{\mathrm{q}}} \dd\mbf x \nx[-r] \rho_{\Phi^{\mathrm{q}} \gamma \Phi^{\mathrm{q}}} (\mbf x) \\
&\leq \|\rho_{\Phi^{\mathrm{q}} \gamma \Phi^{\mathrm{q}}}\|_{1+\frac 2 d} \nor{\nx[-r]\,\chrf_{\Omega^{\mathrm{q}}}}_{1+\frac d 2} .
\end{align*}

Computing
\begin{align*}
\nor{\nx[-r]\,\chrf_{\Omega^{\mathrm{q}}}}_{1+\frac d 2} &= C \big( h^{\alpha (-r (1+\frac d 2)+d)} \big)^{\oolt} \\ &= C h^{\alpha (\scr - r )} \end{align*}
 and using the bound from Lem\/ma~\ref{lem:traceest}(ii)
\[ \|\rho_{\Phi^{\mathrm{q}} \gamma \Phi^{\mathrm{q}}}\|_{1+ \frac 2 d} \leq \scases{  \bigo (h^{-d})  }{  \bigo (h^{-d} \md{\log h}^{\scrh} )  }{  \bigo ( h^{-d - \frac{2  \scr  (s-\scr)}{(2-s)(2-\scr )} } )   }   , \] the claim follows. 
\end{proof}

\begin{rem*}If $s \leq \scr$, the integral error is always less critical than the quantum error. If $s > \scr$, the integral error is less critical than the quantum error whenever
\[ \alpha \left( \scr - \frac d 2 (s - 2) \right) \leq \frac {2 \scr}{2 - \scr } \frac{s - \scr }{2 - s} , \]
which is equivalent to
\[ \alpha \leq \frac{2}{2- s} . \]
\end{rem*}

\subsection{Semiclassical analysis}

\subsubsection{Convolution of $C^1$ and $C^2$ functions with radial kernels}
\begin{lem}[Approximation of $C^1$ functions\footnote{As usual, $f \in C^m (\overline{\Omega + \ball_\tau })$ means that $f \in C^m ((\overline{\Omega + \ball_\tau } )^\circ)$ and its partial derivatives $\partial^{\mbf \alpha} f $ up to order $\vert \mbf \alpha \vert \leq m$ have continuous extensions to $\overline{\Omega + \ball_\tau }$.}]
\label{lem:nconvest}
Let $\Omega \subset \mathbb R^d$ and $f \in C^1 ( \overline{\Omega + \ball_\tau })$. Let $0 \leq g \in C_{\mathrm c}^\infty (\mathbb R^d )$ be a radial function with support in the closed unit ball, normalized to $\nor{g}_{2} = 1$, and let $g_\tau (\mbf x) := \tau^{-d/2} g (\mbf x / \tau )$ for $\tau>0$. Then, there exists a constant $C>0$ such that
\begin{align*} \| ( f_- \star g_\tau^2 - f_- )\,\chrf_\Omega \|_{\infty} &\leq C \tau \nor{\nabla f\;\chrf_{\Omega+\ball_\tau}}_{\infty } . 
\end{align*}
\end{lem} 
\begin{lem}[Approximation of $C^2$ functions]
\label{lem:convest}
Let $\Omega \subset \mathbb R^d$ and $f \in C^2 ( \overline{\Omega + \ball_\tau })$. Let $0 \leq g \in C_{\mathrm c}^\infty (\mathbb R^d )$ be a radial function with support in the closed unit ball, normalized to $\nor{g}_{2} = 1$, and let $g_\tau (\mbf x) := \tau^{-d/2} g (\mbf x / \tau )$ for $\tau>0$. Then, there exists a constant $C>0$ such that
\begin{align*} \| ( f \star g_\tau^2 - f )\,\chrf_\Omega \|_{\infty}  &\leq C \tau^2 \nor{\hess f\;\chrf_{\Omega+\ball_\tau}}_{\infty} . 
\end{align*}
\end{lem}

\begin{proof}[Proof of Lemma~\ref{lem:nconvest}]
Let us first assume that $f \leq 0$. Then, by the fundamental theorem of calculus,
\begin{align*} f(\mbf y) - f(\mbf x) = \int_0^1 \dd \lambda\,(\mbf y - \mbf x) \cdot \nabla f(\lambda \mbf y + (1-\lambda ) \mbf x ) . \end{align*}
Therefore,
\[
(f\star g_\tau^2 - f)(\mbf x) = \int_0^1 \dd \lambda\,\int \dd\mbf y\,g_\tau^2(\mbf y - \mbf x)\,(\mbf y - \mbf x) \cdot \nabla f (\lambda \mbf y + (1-\lambda ) \mbf x).
\]
Thus,
\[
\vert (f\star g_\tau^2 - f)(\mbf x) \vert \leq \nor{\nabla f\,\chrf_{\ball_\tau(\mbf x)}}_{\infty } \int \dd\mbf y g_\tau^2(\mbf y - \mbf x)\,\md{\mbf y - \mbf x}  .
\]
The integral evalues to $ C \tau$ for a constant $C$, hence the claim follows.

If $f \leq 0$ is not fulfilled, $f_-$ is not a differentiable function. However, the function
\begin{align*}
\varphi:\;[0,1] &\longrightarrow \mathbb R \\
\lambda &\longmapsto - f_- (\mbf \gamma(\lambda)),
\end{align*}
where $\mbf \gamma (\lambda ) := \lambda \mbf y + (1-\lambda) \mbf x,$
is still absolutely continuous, and with the notation
\begin{align*}
\mathfrak N := (f \circ \mbf \gamma )^{-1} (\mathbb R_{< 0}) ,
\end{align*} its derivative is given at all $\lambda \in [0,1] \setminus \partial \mathfrak N$ by
\begin{align*}
\varphi'(\lambda ) = \LEFTRIGHT\lbrace. { \begin{array}{ll} (\mbf y - \mbf x) \cdot \nabla f(\mbf \gamma (\lambda ) ) :& \lambda \in \mathfrak N \\
0:& \lambda \notin \mathfrak N \end{array} } ,
\end{align*}
and since $\partial \mathfrak N$ is a Lebesgue nullset, the proof still applies.
\end{proof}

\begin{proof}[Proof of Lemma~\ref{lem:convest}]
By the Taylor formula,
\begin{align*} f(\mbf y) - f(\mbf x) &= (\mbf y - \mbf x) \cdot \nabla f(\mbf x ) + \\ &\qquad + \int_0^1 \dd \lambda\,\lambda\,(\mbf y - \mbf x) \cdot \left( \hess_{\lambda \mbf y + (1-\lambda ) \mbf x} f \right) (\mbf y - \mbf x) . \end{align*}
Therefore,
\begin{multline*}
(f\star g_\tau^2 - f)(\mbf x) = \int \dd\mbf y\,g_\tau^2(\mbf y - \mbf x)\,(\mbf y-\mbf x)\cdot \nabla f(\mbf x) + \\
 + \int_0^1 \dd \lambda\,\lambda\int \dd\mbf y g_\tau^2(\mbf y - \mbf x)\,(\mbf y - \mbf x) \cdot \left(\hess_{\lambda \mbf y + (1-\lambda ) \mbf x} f \right) (\mbf y - \mbf x) .
\end{multline*}
But due to the radial symmetry of $g$, the first-order term vanishes identically. Thus,
\[
\vert (f\star g_\tau^2 - f)(\mbf x) \vert \leq \frac 1 2 \nor{\hess f\,\chrf_{\ball_\tau (\mbf x)}}_{\infty} \int \dd\mbf y\,g_\tau^2(\mbf y - \mbf x)\,\md{\mbf y - \mbf x}^2  .
\]
The integral evalues to $ C \tau^2$ for a constant $C$, hence the claim follows.
\end{proof}

We now apply this lemma to the context $\Omega = \osc n$. We put
\[ \tau_n := h^{\beta_n } . \]
Since we will encounter integrals over $\osc n + \ball_{\tau_n}$, we shall require a priori that the effect of $\tau_n$ be negligible, i.e.
\[ \beta_n> \theta_n . \]
For the sake of simplicity of notation, we define
\begin{align*}
s_n := \LEFTRIGHT\lbrace. { \begin{array}{ll} s :& \text{if}\;n \geq 1 \\
 S :& \text{if}\;n \leq 0 \end{array}} . 
\end{align*}

\begin{cor} \label{cor:cone}
There is a constant $A>0$ such that in the semiclassical zones
\begin{align*}\nor{\big( V_- \star g_{\tau_n}^2 - V_- \big)\,\chrf_{\osc n}}_{\infty} &= \bigo (h^{\beta_n -\theta_n (1 + s_n ) - A \depsilon }) .
\end{align*}
The same bounds apply for $\| ( V \star g_{\tau_n}^2 - V ) \,\chrf_{\osc n} \|_{\infty}$.
\end{cor}

Furthermore, we shall be in the need for a bound on the convolution approximation of $V_-^{d/2}$. Note that
\[ V_-^{\;\frac d 2 } = - \tfrac d 2 V_-^{\;\frac d 2 - 1}\,\nabla V . \]
Therefore, the gradient of $V_-^{d/2}$ is bounded by
\[
\vert \nabla V_-^{\;\frac d 2} \vert \leq C \nx[-\frac d 2 s_n - 1],
\]
This provides us with the following bound.
\begin{cor}
\label{cor:bpow}
In all semiclassical zones,
\begin{align*}
\big\| \big( V_-^{\;\frac d 2} \star g_{\tau_n}^2 - V_-^{\;\frac d 2} \big)\,\chrf_{\osc n} \big\|_{\infty} &= \bigo (h^{\beta_n - \theta_n - \frac d 2 s_n \theta_n }) .
\end{align*}
\end{cor}

For the localization functions, we apply the stronger $C^2$ estimate.

\begin{cor}\label{cor:local} There is a constant $C>0$ such that, for every $n$,
\[ \md{(\psc n)^2 \star g_\tau^2 - (\psc n)^2 } (\mbf x) \leq C h^{2(\beta_n - \theta_n)}\,\chrf_{\osc n + \ball_{\tau_n}} (\mbf x) . \]
\end{cor}

\subsubsection{Semiclassical analysis in the intermediate regions}
We apply the coherent states technique to the regions $\osc n$.

\begin{lem}[Semiclassical error without optimization]
Let $d = 2, 3$. In all inner semiclassical zones $\osc n$,
\begin{align*} &\tr [\psc n (-h^2 \Delta + V ) \psc n ]_- = \lcl h^{-d} \int \dd\mbf u\,V_-^{1+\frac d 2} (\psc n )^2 + R ,
 \end{align*}
where the semiclassical error $R$ is:
\begin{center}
 \begin{tabular}{lll}
\toprule
           &   \multicolumn{1}{c}{$d = 2$}                                           &  \multicolumn{1}{c}{$d = 3$}   \\
   \midrule
$\theta_{n-1} \geq 0$ &   $\bigo\big(h^{-2 + \beta_n + \theta_n (1 - 2 s) - A \epsilon}$        &  $\bigo\big(h^{-3 + \beta_n + \theta_n (2 - \frac 5 2 s) - A \epsilon}$   \\
           &   $+\;h^{-2 + 2 ( 1 - \beta_n) + \theta_n (2 - s) - A \epsilon}$        &  $+\;h^{-3 + 2 ( 1 - \beta_n) + \theta_n (3 - \frac 3 2 s) - A \epsilon}$   \\
           &   $+\;h^{-2 + 2 \beta_n - \theta_n (1 + s) - A \epsilon} \big)$         &  $+\;h^{-3 + 2 \beta_n - \theta_n (1 + \frac 5 2 s) - A \epsilon} \big)$   \\[.5em]
   
$\theta_{n-1} < 0$ &   $\bigo\big(h^{-2 + \beta_n + \theta_n (1 - S) - A \epsilon}$          &  $\bigo\big(h^{-3 + \beta_n + \theta_n (2 - \frac 3 2 S) - A \epsilon}$   \\
           &   $+\;h^{-2 + 2 ( 1 - \beta_n) + \theta_n (2 - S) - A \epsilon}$        &  $+\;h^{-3 + 2 ( 1 - \beta_n) + \theta_n (3 - \frac 3 2 S) - A \epsilon}$   \\
           &   $+\;h^{-2 + 2 \beta_n - \theta_n S - A \epsilon} \big)$               &  $+\;h^{-3 + 2 \beta_n - \theta_n (1 + \frac 5 2 S) - A \epsilon} \big)$    \\
\bottomrule
\end{tabular}
\end{center}
\end{lem}

With the optimal choice for $\beta_n$, this implies:

\begin{cor}[Semiclassical error]
\label{cor:scerror}
Let $d = 2, 3$. In all inner semiclassical zones $\osc n$ with
$ \theta_n \leq \frac 2 {2 - s} $,
\begin{align*} &\tr [\psc n (-h^2 \Delta + V ) \psc n ]_- = \lcl h^{-d} \int \dd\mbf u\,V_-^{1+\frac d 2} (\psc n )^2 + R ,
 \end{align*}
where the semiclassical error $R$ is
\begin{center}
 \begin{tabular}{lll}
\toprule
           &   \multicolumn{1}{c}{$d = 2$}                                   &  \multicolumn{1}{c}{$d = 3$}   \\
   \midrule
$\theta_n \geq \frac{2}{8 - s}$  &       & $\bigo\big( h^{-2 + \theta_n (1 - 2 s) - A \epsilon }  \big)$ \\[.5em]
$\theta_{n-1} \geq 0$ &   $\bigo\big(h^{-\frac 4 3 + \theta_n (\frac 4 3 - \frac 5 3 s) - A \epsilon} \big)$        &  $\bigo\big(h^{-\frac 7 3 + \theta_n (\frac 7 3  - \frac {13} 6 s) - A \epsilon}\big)$   \\[.5em]
$\theta_{n-1} < 0$ &   $\bigo\big(h^{-\frac 4 3 + \theta_n (\frac 4 3 - S) - A \epsilon}\big)$          &  $\bigo\big(h^{-\frac 7 3 + \theta_n (\frac 7 3 - \frac 3 2 S) - A \epsilon} \big)$   \\
\bottomrule
\end{tabular}
\end{center}
\end{cor}

\newcommand{\gam}{\psc n \gamma \psc n}
\newcommand{\Tr}{\tr}
\newcommand{\taun}{{\tau_n}}
\newcommand{\sn}{s_n}
\newcommand{\tn}{{\nu_n + s_n}}
\newcommand{\Tn}{{\nu_n + s_n}}
\begin{proof}[Proof of the lemma] The proof is valid for all dimensions $d \geq 2$. However, for each of the three error terms in the proof of the upper bound there is a comparable\footnote{By `comparable' we mean that one error term dominates over the other independently of the value of $\beta_n$.} error term in proof of the lower bound, and also at one point there is a distinction between inner and outer zones, and which one dominates depends on $d$.

Throughout the proof, we use the notation
\begin{align*}
 \mathfrak M &:= \lbrace (\mbf u, \mbf p ) : h^2 \vert\mbf{p}\vert^2 + V(\mbf u) <0 \rbrace \\
 \mathfrak M_{\mbf u} &:= \lbrace \mbf p : h^2 \vert\mbf{p}\vert^2 + V(\mbf u) <0 \rbrace .
\end{align*}
Also, we put w.l.o.g. $V(\mbf x) := 0$ outside $\osc n + \ball_{\tau_n}$.

We prove the lower bound first. Putting
\[ \gamma := \chrf_{(-\infty, 0)} [\psc n (-h^2 \Delta + V  ) \psc n ] , \]
we get
\begin{align*}
&- \tr [\psc n (-h^2 \Delta + V ) \psc n ]_- = \tr [(-h^2 \Delta + V  )\,\psc n \gamma \psc n ] \\
&\qquad= \ootp d \ipu p u\,[h^2 \vert\mbf{p}\vert^2 + V(\mbf u)  ]\,\tr (\prf {\tau_n} u p \psc n \gamma \psc n )\;+ \\
&\qquad\qquad\qquad + \tr [(V - V\star g_{\tau_n}^2 - C h^{2(1-\beta_n )} )\,\psc n \gamma \psc n ] , \tag{$*$}
\end{align*}
where $C = \nor{\nabla g}_{2}^2 $.

Lemma~\ref{lem:traceest} yields
\begin{align*}
\tr (\gam ) &= \bigo \big(  C h^{-d + \theta_n (- \frac d 2 s_n + d ) - A \depsilon } \big) .
\end{align*}
This bound, along with the bound from Corollary~\ref{cor:cone}, yield
\begin{align*}
\tr [(V - V\star g_{\tau_n}^2 )\,\gam ] &= \bigo (h^{-d + \beta_n + \theta_n (- \frac d 2 s_n + d -1 - s_n) - A \depsilon }) ,      \\
\tr [ C h^{2 (1- \beta_n )}\,\gam ] &= \bigo (h^{- d + 2 (1 - \beta_n ) + \theta_n (- \frac d 2 s_n + d) - A \depsilon }) .
\end{align*}
To bound the first term in eqn.~($*$), we use
\[  0 \leq \tr (\prf \taun u p \gam ) \leq \tr (\prf \taun u p (\psc n)^2 ) = ( (\psc n)^2 \star g_\taun^2 )(\mbf u) ,  \]
hence
\begin{align*}
&\ootp d \ipu p u\,[h^2 \vert\mbf{p}\vert^2 + V(\mbf u) ]\,\tr (\prf \taun u p \gam ) \\
&\qquad\qquad \geq - \ootp d \ipu[\mathfrak M] p u\,[h^2 \vert\mbf{p}\vert^2 + V(\mbf u)  ]_-\,((\psc n)^2 \star g_\taun^2 )(\mbf u)    \\
&\qquad\qquad \geq - \lcl h^{-d} \int\dd\mbf u\,[V(\mbf u) ]_-^{1+\frac d 2} (\psc n)^2 (\mbf u)\;- \\
&\qquad\qquad\hphantom{\;\geq\;}{- \lcl} h^{-d}  \int\dd\mbf u\,[V(\mbf u)  ]_-^{1+\frac d 2} \md{(\psc n)^2 - (\psc n)^2 \star g_\taun^2 } (\mbf u).
\end{align*}
By virtue of Corollary~\ref{cor:local} and Hölder's inequality, the last term of the above inequality can be bounded by
\[
\lcl h^{-d}   \int\dd\mbf u\,[V(\mbf u)  ]_-^{1+\frac d 2} \md{(\psc n)^2 - (\psc n)^2 \star g_\taun^2 } (\mbf u) \leq C\,h^{- d  + 2(\beta_n - \theta_n )} \nor{V_-}_{\infty}^{1+\frac d 2 }\,\vert \osc n \vert .
\]
We compute
\[  \nor{V _-}_{\infty}^{1+\frac d 2 }\,\vert \osc n \vert = \bigo (h^{\theta_n (d - \sn (1 + \frac d 2 )) - A \depsilon}) .
 \]
Therefore, the above term becomes
\[
\lcl h^{-d}   \int\dd\mbf u\,[V(\mbf u) ]_-^{1+\frac d 2} \md{(\psc n)^2 - (\psc n)^2 \star g_\taun^2 } (\mbf u) = \\
= \bigo (h^{- d + 2 \beta_n - \theta_n (1 + \sn (d - 2 + \tfrac d 2 )) - A \depsilon}) .
\]
For the upper bound, we choose
\[ \gamma := \ootp d \ipu[\mathfrak M] p u\,\prf \taun u p .  \]
Then,
\begin{align*}
&-\tr [\psc n (-h^2 \Delta + V  ) \psc n]_- \leq \tr [ (-h^2\Delta + V )\,\gam ] \\
&\qquad\qquad= \ootp d \ipu[\mathfrak M] p u\,\braf \taun u p \psc n (-h^2 \Delta + V  ) \psc n \ketf \taun u p \\
&\qquad\qquad= \ootp d \ipu[\mathfrak M] p u\,\brawv p  -(g_\taun \psc n)^2 \tfrac{h^2}{2}\Delta - \tfrac{h^2}{2}\Delta (g_\taun \psc n)^2\;+ \\[-1em]
&\qquad\qquad\qquad\qquad\qquad +\; h^2 \md{\nabla (g_\taun \psc n)}^2 + (g_\taun \psc n)^2 V \ketwv p   \\
&\qquad\qquad= \ootp d \ipu[\mathfrak M] p u\,\Big(  h^2 \vert\mbf{p}\vert^2 ((\psc n)^2 \star g_\taun^2)(\mbf u)\;+ \\[-1em]
&\qquad\qquad\qquad\qquad\qquad +\;h^2 \int \dd\mbf x\,\md{\nabla (g_\taun \psc n )(\mbf x)}^2\;+ \\
&\qquad\qquad\qquad\qquad\qquad +\;\big( (\psc n)^2 V  \star g_\taun^2 \big)(\mbf u) \Big) \\
&\qquad\qquad= \lkin h^{-d} \int \dd\mbf u\,[V(\mbf u) ]_-^{1+\frac d 2} ((\psc n)^2 \star g_\taun^2 )(\mbf u)\;+ \\
&\qquad\qquad\qquad +\; \lpot h^{2-d} \ipu u x\,[V(\mbf x) ]_-^{\;\frac d 2}\,\md{\nabla (g_\taun \psc n )(\mbf x )}^2\; -\\
&\qquad\qquad\qquad -\;\lpot h^{-d} \int\dd\mbf u\,[V(\mbf u) ]_-^{\;\frac d 2} \big( (\psc n)^2 V  \star g_\taun^2 \big)(\mbf u)  \\
&\qquad\qquad= - \lcl h^{-d} \int \dd\mbf u [V(\mbf u) ]_-^{1+\frac d 2} (\psc n(\mbf u))^2\;+ \\
&\qquad\qquad\qquad +\;\lkin h^{-d} \int \dd\mbf u\,[V(\mbf u)  ]_-^{1+\frac d 2}\;\times \\*
&\qquad\qquad\qquad\qquad\qquad\qquad\times\; \big[ (\psc n)^2 \star g_\taun^2 - (\psc n )^2  \big] (\mbf u)\; + \qquad &\smash{\raisebox{1em}{$\Bigg\rbrace \;I_1$}} \\
&\qquad\qquad\qquad +\;\lpot h^{-d} \int\dd\mbf u\,V(\mbf u)\,(\psc n (\mbf u))^2 \;\times \\*
&\qquad\qquad\qquad\qquad\qquad\qquad\times\;\big(V _-^{\;\frac d 2} \star g_\taun^2  - V_-^{\;\frac d 2} \big) (\mbf u)\; + &\smash{\raisebox{1em}{$\Bigg\rbrace \;I_2$}} \\
&\qquad\qquad\qquad +\;\lpot h^{2-d} \ipu u x\,[V(\mbf u)]_-^{\;\frac d 2}\,\md{\nabla (g_\taun \psc n )(\mbf x)}^2 \; , \qquad  &\Big\rbrace \;I_3
\end{align*}
where we used that
\begin{multline*}
 \int\dd\mbf u\,[V(\mbf u)]_-^{\;\frac d 2} \big[ V(\mbf u)\,(\psc n)^2 (\mbf u) - ((\psc n)^2 V) \star g_\taun^2 (\mbf u) \big] = \\
= \int\dd\mbf u\,V(\mbf u)\,(\psc n (\mbf u))^2 \big[ V_-^{\;\frac d 2} \star g_\tau^2 (\mbf u) - [V (\mbf u) ]_-^{\;\frac d 2}\big] .
\end{multline*}
$I_1, I_2$ and $I_3$ are error terms. The term
\[
I_1 = \bigo (h^{- d + 2 \beta_n - \theta_n (d-2 + \sn (1 + \frac d 2 )) - A \depsilon}) \]
has already been treated. The term
$I_2$
we bound by
\begin{align*}
& \lpot h^{-d} \int\dd\mbf u\,V(\mbf u)\,(\psc n (\mbf u))^2\,\big[ V_-^{\;\frac d 2} \star g_\taun^2  - V_-^{\;\frac d 2} \big] (\mbf u) \leq \\
&\qquad\qquad \leq C\,h^{-d}\; \vert \osc n + \ball_\taun \vert  \; \nor{V }_{\infty}  \big\|  V_-^{\;\frac d 2} \star g_\taun^2  - V_-^{\;\frac d 2}  \big\|_{\infty}   \\
&\qquad\qquad = \LEFTRIGHT\lbrace. {\begin{array}{ll}
   \bigo \big( h^{- d + \theta_n ( d - (1+\frac d 2) s_n ) +  \beta_n - \theta_n  - A \depsilon} \big) : & \text{if}\;n \geq 1    \\
   \bigo \big( h^{- d + \theta_n ( d - \frac d 2 s_n )\hphantom{(1+)} +  \beta_n - \theta_n  - A \depsilon} \big) : & \text{if}\;n \leq 0
   \end{array} }.
\end{align*}
The distinction $n \geq 1$ vs. $n \leq 0$ comes from the term $\| V \|_\infty$. For $n \leq 0$, we are using the boundedness of $\vert V \vert$ on $\lbrace \nx \geq 1 \rbrace$.

Finally, the double integral $I_3$ can be bounded as follows:
\begin{align*} \int \dd\mbf x\,\md{\nabla (g_\taun (\psc n)^2 )(\mbf x)}^2 &\leq \md{\ball_\taun}\,C \big( h^{- \beta_n \frac d 2} h^{-\theta_n} + h^{-\beta_n \left( 1 + \frac d 2 \right)} \big)^2 = \bigo ( h^{-2\beta_n } ) , \\
\smashoperator{\int_{\osc n + \ball_\taun}} \dd\mbf u\,[V(\mbf u)  ]_-^{\;\frac d 2} &\leq \md{\osc n + \ball_\taun}\;\big\| V_- \big\|_\infty^{\frac d 2}  \\[-1em]
&= \bigo (h^{\theta_n ( d - s_n \frac d 2 ) - A \depsilon}) ,  \end{align*}
therefore
\[ \lpot h^{2-d} \ipu u x\,[V(\mbf u)]_-^{\;\frac d 2}\,\md{\nabla (g_\taun \psc n )(\mbf x)}^2 = \bigo (h^{- d + 2 + \theta_n ( d - s_n \frac d 2 ) - 2 \beta_n  - A \depsilon}) . \qedhere \]
\end{proof}

\begin{proof}[Proof of the corollary]
The optimal choices for $\beta_n$ are:
\begin{center}
 \begin{tabular}{lll}
\toprule
           &   \multicolumn{1}{c}{$d = 2$}                                   &  \multicolumn{1}{c}{$d = 3$}   \\
   \midrule
$\theta_n \geq \frac{2}{8 - s}$  &       & $\frac 1 2 + \theta_n \frac{4+s}{4} $ \\[.5em]
$\theta_{n-1} \geq 0$ &   $\frac{2}{3}+\theta_n \frac{1+s}{3}$        &  $\frac 2 3 + \theta_n \frac{1+s}{3}$   \\[.5em]
$\theta_{n-1} < 0$ &   $\frac 2 3 + \frac{1} 3 \theta_n$          &  $\frac 2 3 + \frac{1} 3 \theta_n$   \\
\bottomrule
\end{tabular}
\end{center}
\end{proof}

\subsection{Optimization of the quantum zone}

It remains to fix $\alpha$. Three error sources have to be considered: The localization error, the quantum error and the semiclassical error. With a higher value of $\alpha$ (i.e. a smaller quantum zone), we get a better quantum error, but worse localization and semiclassical errors.

The subsequent lemmas already capture the essence of the main theorems. However, note that the $\eta^*$ to be defined in the following is not guaranteed to be positive, i.e. the error terms are not guaranteed to be of order $\lio (h^{-d})$.

\begin{lem}[Two-dimensional case]
\label{lem:m2c1}
Let $d = 2$ and assume $1 \leq s < 2$, $r < 2 - s$. Let $\omega > 0$. Then,
\begin{align*}
&\tr [- h^2 \Delta + V_1 ]_- - \tr [- h^2 \Delta + V_2  ]_- \\*
&\quad= \frac{1}{8 \pi h^2} \,\int \dd\mbf x\;\big( [ V_1(\mbf x) ]_-^{2} - [ V_2 (\mbf x ) ]_-^{2 } \big) \;- \\*
&\quad\quad- \frac{1}{8 \pi h^2} \,\int \dd\mbf x\;\big( [ V_1(\mbf x) ]_-^{2} - [ V_2 (\mbf x ) ]_-^{2 } \big)\,\Big(\sum_{n < - \omega / \epsilon } (\psc n)^2 (\mbf x)\Big)\;+ \\*
&\quad\quad+ \sum_{n < -\omega / \epsilon } \big(\tr [\psc n (- h^2 \Delta + V_1 ) \psc n]_- - \tr [\psc n (-h^2 \Delta+ V_2 ) \psc n ]_- \big)\;+ \\*
&\quad\quad+  \bigo (h^{-2 + \eta^* - A \epsilon }),
\end{align*}
where $\eta^* = \min \lbrace \eta_{\mathrm{sc}}, \eta_{\mathrm{loc}}, \eta_{\mathrm{cutoff}} \rbrace$, with
\begin{align*}
\eta_{\mathrm{sc}} &= -\frac{2 (s - 1)}{2 - s} + \frac{2 (1 - r) (2 s - 1)}{(2 - s)(5 s - 3 r - 1)}  \\*
\eta_{\mathrm{loc}} &= - \frac{2 (s - 1)}{2 - s} + \frac{2 (1 - r) (2 s - 3)}{(2 - s)(2 - r)} \\*
\eta_{\mathrm{cutoff}} &= \frac 2 3 - \Big(\frac 4 3 - S \Big)\,\omega .
\end{align*}
\end{lem}

\begin{proof}
The term $\eta_{\mathrm{cutoff}}$ comes from the semiclassical error in $\osc {\lceil -\omega / \epsilon \rceil}$. The term $\eta_{\mathrm{sc}}$ is obtained by choosing $\alpha$ such that the errors of the innermost semiclassical zone and the quantum zone are of the same order, which yields
\[ \alpha_{\mathrm{sc}} = \frac{2 (2 s - 1)}{(2- s)(5 s - 3 r - 1)} , \]
and the term $\eta_{\mathrm{loc}}$ is obtained from the $\alpha$ that puts localization error and quantum error in equilibrium, namely
\[ \alpha_{\mathrm{loc}} = \frac{2 (3 - 2 s)}{(2 - s)(2 - r)} . \]
Both fulfill
\[ \alpha_{\mathrm{loc}} , \alpha_{\mathrm{sc}} \leq \frac{2}{2 - s} . \]
If for $\alpha := \alpha_{\mathrm{sc}}$ the semiclassical/quantum error $\bigo (h^{- \eta_{\mathrm{sc}} - A \epsilon})$ is of higher order than the localization error, choose this $\alpha$, otherwise choose $\alpha := \alpha_{\mathrm{loc}}$; in the latter case, the semiclassical error is of lower order than the quantum/localization error, as follows elementarily.

Lemma~\ref{lem:localization}, Lemma~\ref{lem:compone} and Corollary~\ref{cor:scerror} imply
\begin{align*}
&\tr [- h^2 \Delta + V_1  ]_- - \tr [- h^2 \Delta + V_2   ]_- \\*
&\quad = \frac{1}{8 \pi h^2} \sum_{n \geq - \omega / \epsilon } \int \dd\mbf x\,\big( [V_1(\mbf x)]_-^2 - [V_2(\mbf x) ]_-^2 \big)\,(\psc n)^2(\mbf x)\;+\\*
&\quad\quad+\sum_{n < -\omega / \epsilon } \big(\tr [\psc n (- h^2 \Delta + V_1 ) \psc n]_- - \tr [\psc n (-h^2 \Delta+ V_2 ) \psc n ]_- \big)\; +\\*
&\quad\quad+\;\bigo(h^{-2 + \eta^* - A \epsilon}) .
\end{align*}
Furthermore, by virtue of Corollary~\ref{cor:quantum-int}, we write
\begin{align*}
& \frac{1}{8 \pi h^2} \int \dd\mbf x\,\big( [V_1(\mbf x) ]_-^2 - [V_2(\mbf x) ]_-^2 \big)\,\Big( \sum_{n > - \omega / \epsilon }(\psc n)^2(\mbf x) \Big) \\*
&\quad = \frac{1}{8 \pi h^2} \int \dd\mbf x\,\big( [V_1(\mbf x) ]_-^2 - [V_2(\mbf x)]_-^2 \big) + \bigo (h^{-2+\eta - A \epsilon})\; - \\*
&\quad\quad - \frac{1}{8 \pi h^2} \,\int \dd\mbf x\;\big( [ V_1(\mbf x) ]_-^{2} - [ V_2 (\mbf x )]_-^{2 } \big)\,\Big(\sum_{n < - \omega / \epsilon } (\psc n)^2 (\mbf x) \Big) .
\end{align*}
This proves the lemma.
\end{proof}

\begin{lem}[Three-dimensional case]
\label{lem:m3}
Let $d = 3$ and assume $1 \leq s < \frac{43 - \sqrt{769}}{10} \simeq 1.5269$\footnote{$1.5269$ is a truncated value.} and $r < \min \lbrace s, \frac 3 2 (2 - s) \rbrace$.
Let $\omega > 0$. Then,
\begin{align*}
&\tr [- h^2 \Delta + V_1 ]_- - \tr [- h^2 \Delta + V_2  ]_- \\*
&\quad= \frac{1}{15 \pi^2 h^3} \,\int \dd\mbf x\;\big( [ V_1(\mbf x) ]_-^{\;\frac 5 2} - [ V_2 (\mbf x ) ]_-^{\;\frac 5 2 } \big) \;- \\*
&\quad\quad- \frac{1}{15 \pi^2 h^3} \,\int \dd\mbf x\;\big( [ V_1(\mbf x) ]_-^{\;\frac 5 2} - [ V_2 (\mbf x ) ]_-^{\;\frac 5 2  } \big)\,\Big( \sum_{n < - \omega / \epsilon } (\psc n)^2 (\mbf x) \Big) \;+ \\*
&\quad\quad+ \sum_{n < -\omega / \epsilon } \big(\tr [\psc n (- h^2 \Delta + V_1 ) \psc n]_- - \tr [\psc n (-h^2 \Delta+ V_2 ) \psc n ]_- \big)\;+ \\*
&\quad\quad+  \bigo (h^{-3+\eta^* - A \epsilon }) ,
\end{align*}
If $1 \leq s \leq \frac 6 5$, $\eta^* = \min \lbrace \eta_{\mathrm{loc}}, \eta_{\mathrm{cutoff}} \rbrace$ with
\begin{align*}
\eta_{\mathrm{loc}} &= 2 - \frac{8}{5 (2-r)} > 0 \\
\eta_{\mathrm{cutoff}} &= \frac 2 3 - \Big(\frac 7 3 - \frac 3 2 S \Big)\,\omega .
\end{align*}
If $s > \frac 6 5$, $\eta^* = \min \lbrace \eta_{\mathrm{sc}}, \eta_{\mathrm{loc}}, \eta_{\mathrm{cutoff}} \rbrace$ with
\begin{align*}
\eta_{\mathrm{sc}} &= 2 -\frac{5 (s - \frac 6 5)}{2-s} + \frac{8 (5 s - 4)}{5 (2 - s) (10 s - 5 r + 1)} \\*
\eta_{\mathrm{loc}} &= 2 - \frac{5 (s - \frac 6 5) }{2-s} + \frac{16 (8 - 5 s)}{25 (2 - s) (2 - r)} \\*
\eta_{\mathrm{cutoff}} &= \frac 2 3 - \Big(\frac 7 3 - \frac 3 2 S \Big)\,\omega .
\end{align*}
\end{lem}

\begin{proof}
The proof is analogous, with
\begin{align*}
\alpha_{\mathrm{sc}} &= \frac{5}{10 s - 5 r + 1} \\
\alpha_{\mathrm{loc}} &= \frac{2}{2 - r}
\end{align*}
for $s \leq \frac 6 5$, in which case it is easily verified that $0< \eta_{\mathrm{loc}} \leq \eta_{\mathrm{sc}}$ for all admissible values of $r$, and with
\begin{align*}
\alpha_{\mathrm{sc}} &= \frac{2 (5 s - 4)}{(2 - s) (10 s -5 r + 1)} \\
\alpha_{\mathrm{loc}} &= \frac{4 (8 - 5 s)}{5 (2 - s)(2 - r)}
\end{align*}
for $s > \frac 6 5$. The constraint $s < \frac{43 - \sqrt{769}}{10}$ ensures that, for arbitrary positive $r$,
\[ \alpha_{\mathrm{sc}} , \alpha_{\mathrm{loc}} \geq \frac{2}{8 - s} . \]
\end{proof}

\subsection{The outer zones}
\label{sect:outerzones}
We saw that, while the localization error becomes less significant with decreasing zone index $n$ under mild conditions, the semiclassical error grows (unless $S$ is very large). Therefore, we have to cut it off and apply the coherent state analysis only to zones with exponents $\theta_n \geq - \omega$. We have to account for the error terms
\begin{multline*}
h^{-d}\,\Bigg\vert \int \dd\mbf x\; \big( [ V_1(\mbf x) ]_-^{1 + \frac d 2} - [ V_2 (\mbf x ) ]_-^{1 + \frac d 2 } \big)  \sum_{n < - \omega / \epsilon } (\psc n)^2 (\mbf x) \Bigg\vert \\*
\leq h^{-d} \smashoperator{\int\limits_{\nx \geq h^{- \omega }}}\dd\mbf x\; \big\vert [ V_1(\mbf x) ]_-^{1 + \frac d 2} - [ V_2 (\mbf x ) ]_-^{1 + \frac d 2 } \big\vert
\end{multline*}
and
\[ \sum_{n < -\omega / \epsilon } \big(\tr [\psc n (- h^2 \Delta + V_1  ) \psc n]_- - \tr [ \psc n (-h^2 \Delta+ V_2  ) \psc n ]_- \big) , \]
as well as for the outermost semiclassical error
\[ \bigo (h^{-d + \eta_{\mathrm{cutoff}}-A \epsilon})\quad\text{with}\quad \eta_{\mathrm{cutoff}} = \LEFTRIGHT\lbrace. { \begin{array}{ll}
\frac 2 3 - \big(\frac 4 3 - S \big)\,\omega : & \text{for}\;d = 2 \\
\frac 2 3 - \big(\frac 7 3 - \frac 3 2 S \big)\,\omega :& \text{for}\; d = 3 .
\end{array}} \]
We bound
\[ h^{-d} \smashoperator{\int\limits_{\nx \geq h^{-\omega}}}\dd\mbf x\,[V(\mbf x)]_-^{1+\frac d 2} = \bigo (h^{- d + \omega d \frac{S - \scr}{\scr }}) \]
and, by the Lieb--Thirring inequality,
\[ \sum_{n < - \omega / \epsilon } \Tr [\psc n (-h^2 \Delta + V  ) \psc n]_- \leq C h^{-d}\smashoperator{\int\limits_{\nx \geq h^{-\omega}}}\dd\mbf x\,[V(\mbf x)]_-^{1+\frac d 2} = \bigo (h^{- d + \omega d \frac{S - \scr}{\scr }}) .  \]
We distinguish two cases. If
\[  S > \LEFTRIGHT\lbrace.{ \begin{array}{ll} \frac 4 3 : & \text{for}\;d = 2  \\ \frac{14} 9 :&\text{for}\;d = 3 ,  \end{array} } \]
then $\omega$ can be chosen arbitrarily large, hence also $\eta_{\mathrm{cutoff}}$ can be made arbitrarily large, and all the errors in this section are of lower order than the other errors. On the other hand, if
\[  S \in \LEFTRIGHT\lbrace. {\begin{array}{ll} (1, \frac 4 3 ] : & \text{for}\;d = 2  \\ ( \frac 6 5 , \frac{14} 9 ] :&\text{for}\;d = 3 ,  \end{array} } \]
$\omega$ must be chosen so as to balance the errors,
\[ \omega d \frac{S - \scr }{\scr } = \eta_{\mathrm{cutoff}} , \]
which is solved for
\[ \omega = \frac{\frac 2 3}{S - \frac 2 3} , \]
and with that choice the error exponent becomes
\[ \eta_{\mathrm{cutoff}} = \LEFTRIGHT\lbrace.{ \begin{array}{ll} \displaystyle \frac 4 3 \frac{S - 1}{S - \frac 2 3}:&\text{for}\; d = 2 \\[1em]
\displaystyle \frac 5 3 \frac{ S - \frac 6 5}{ S - \frac 2 3 }:&\text{for}\; d= 3 . \end{array}   } \]

\subsection{Proofs of the main theorems}
As anticipated earlier, for $\eta < \eta^*$, put
\[  \epsilon := \frac{\eta^* - \eta }{A} .  \]
The main theorems then follow essentially from Lemmas~\ref{lem:m2c1} -- \ref{lem:m3} and the choice of $\omega$ delinated in the previous section. Note that the conditions on $r$ in the main theorems are all stronger than the condition in those lemmas. It only remains to prove the guarantees on the positivity of $\eta^*$. We start with the three-dimensional case, where only the case $s > \frac 6 5$ is non-trivial. Since $\eta_{\mathrm{cutoff}}$ is manifestly positive (and involves only $S$), it is left to ensure that $\eta_{\mathrm{sc}}$ and $\eta_{\mathrm{loc}}$ are positive. By factorization, we obtain the expressions
\begin{align*} \eta_{\mathrm{sc}} &= \frac{175 r s-250 r-350 s^2+425 s+82}{5 (2-s) (10 s - 5 r + 1)} \\
\eta_{\mathrm{loc}} &= \frac{175 r s-250 r-270 s+372}{25 (2 - s) (2 - r)} .  \end{align*}
We get $\eta_{\mathrm{sc}} > 0$ for
\[ r < \frac{-350 s^2+425 s+82}{250-175 s} \qquad \big(\!= 0\;\text{for}\;s = \frac{85 + 3 \sqrt{1313}}{140} \simeq 1.3836 \big), \]
and $\eta_{\mathrm{loc}} > 0$ for
\[ r < \frac{6 (62 - 45 s)}{25 (10 - 7 s)} \qquad \big(\!= 0\;\text{for}\;s = \frac{62}{45} = 1.377\dots \big) . \]
It turns out that the maximal $r$ from the latter condition is, for all $s < \frac{62}{45}$, smaller than that from the former condition. 

For the two-dimensional case,
\begin{align*}
\eta_{\mathrm{sc}} &= \frac{2 (6 - 2 r - 10 s + r s + 5 s^2)}{(2 - s) (5s - 1 - 3r )}  \\
\eta_{\mathrm{loc}} &= \frac{2 (5 - 4 r - 4 s + 3 r s)}{(2-s)(2-r)} .
\end{align*}
We have $\eta_{\mathrm{sc}} > 0$ for
\[ r < \frac{-5 s^2+8 s-2}{2-s} \qquad \big(\!= 0\;\text{for}\;s = \frac{4 + \sqrt 6}{5} \simeq 1.2899  \big) , \]
$\eta_{\mathrm{loc}} > 0$ for
\[ r < \frac{5 - 4 s}{4 - 3 s} \qquad \big(\!= 0\;\text{for}\;s = \frac{5}{4} = 1.25 \big) . \]

\end{document}